\documentclass[11pt]{article}
\usepackage[overload]{empheq}
\usepackage{amsmath,amsthm,amsfonts,amscd,eucal,latexsym,amssymb,mathrsfs,dsfont,cancel,mathtools,physics,tensor,comment,appendix,enumerate,enumitem}
\usepackage{tcolorbox}
\usepackage{epsfig}
\usepackage{simplewick,bm}
\definecolor{hypercolor}{rgb}{0,0.2,0.7}
\usepackage[unicode,pdfencoding=unicode,psdextra,breaklinks,colorlinks,linkcolor=hypercolor,urlcolor=hypercolor,citecolor=hypercolor]{hyperref}
\usepackage{csquotes}
\usepackage{microtype}
\usepackage{mleftright}
\usepackage{soulutf8}
\def\cA{{\mathcal A}}
\def\cB{{\mathcal B}}
\def\cC{{\mathcal C}}

\def\cH{{\mathcal H}}

\def\cM{{\mathcal M}}

\def\mI{\mathscr{I}}

\def\RR{{\mathbb R}}

\def\beq{\begin{eqnarray}}
\def\eeq{\end{eqnarray}}
\def\pa{\partial}
\def\at{\mleft(}               
\def\aq{\mleft[}               

\def\ct{\mright)}              
\def\cq{\mright]}              

\newtheorem{theorem}{Theorem}[section]

\newtheorem{proposition}[theorem]{Proposition}
\newtheorem{lemma}[theorem]{Lemma}
\theoremstyle{definition}
\newtheorem{definition}{Definition}[section]
\newtheorem{remark}{Remark}[section]

\def\operatorT{\mathcal{T}}
\newcommand{\wick}[1]{{:}#1{:}}

\newcommand{\expvalom}[1]{\expval{\wick{#1}}_\omega}

\newcommand{\dif}{\mathop{}\!\mathrm{d}}

\usepackage{cite}

\begin{document}


\par
\bigskip
\LARGE
\noindent
{\bf Existence and uniqueness of solutions of the semiclassical Einstein equation in cosmological models}
\bigskip
\par
\rm
\normalsize


\large
\noindent
{\bf Paolo Meda$^{1,2,a}$}, {\bf Nicola Pinamonti$^{3,2,b}$}, {\bf Daniel Siemssen$^{4,c}$}  \\
\par
\small
\noindent$^1$ Dipartimento di Fisica, Universit\`a di Genova, Italy.
\smallskip

\noindent$^2$ Istituto Nazionale di Fisica Nucleare - Sezione di Genova, Italy.

\smallskip

\noindent$^3$ Dipartimento di Matematica, Universit\`a di Genova, Italy.

\smallskip

\noindent$^4$ Department of Mathematics, University of York, UK.
\smallskip

\noindent E-mail:
$^a$paolo.meda@ge.infn.it,
$^b$pinamont@dima.unige.it,
$^c$daniel.siemssen@york.ac.uk
\normalsize

\par

\rm\normalsize

\rm\normalsize


\par
\bigskip

\rm\normalsize
\noindent {\small Version of \today}

\par
\bigskip

\rm\normalsize

\bigskip

\noindent
\small
{\bf Abstract.}
We prove existence and uniqueness of solutions of the semiclassical Einstein equation in flat  cosmological spacetimes driven by a quantum massive scalar field with arbitrary coupling to the scalar curvature. In the semiclassical approximation, the backreaction of matter to curvature is taken into account by equating the Einstein tensor to the expectation values of the stress-energy tensor in a suitable state. We impose initial conditions for the scale factor at finite time and we show that a regular state for the quantum matter compatible with these initial conditions can be chosen.
Contributions with derivative of the coefficient of the metric higher than the second are present in the expectation values of the stress-energy tensor and the term with the highest derivative appears in a non-local form.
This fact forbids a direct analysis of the semiclassical equation, and in particular, standard recursive approaches to approximate the solution fail to converge. In this paper we show that, after partial integration of the semiclassical Einstein equation in cosmology, the non-local highest derivative appears in the expectation values of the stress-energy tensor through the application of a linear unbounded operator which does not depend on the details of the chosen state. We prove that an inversion formula for this operator can be found, furthermore, the inverse happens to be more regular than the direct operator and it has the form of a retarded product, hence causality is respected.
The found inversion formula applied to the traced Einstein equation has thus the form of a fixed point equation. The proof of local existence and uniqueness of the solution of the semiclassical Einstein equation is then obtained applying the Banach fixed point theorem.

\section{Introduction}

The analysis of the backreaction of linear quantum fields in the context of cosmological spacetimes has been developed in several recent works \cite{Dappiaggi:2008mm, Pinamonti:2010is, Pinamonti:2013wya, Hack:2015zwa, Gottschalk:2018kqt}.
In those works, the quantization of linear fields on curved backgrounds is performed using the algebraic approach (see e.g.~\cite{Haag:1992hx, LibroAQFT}).
According to that paradigm, the first step is the construction of the algebra of observables.
Actually, even if there is no preferred state to choose on a generic curved spacetime, on globally hyperbolic spacetimes the canonical commutation relations of linear fields can be given prescribing the form of the product among the generators of this algebra \cite{Brunetti:2001dx, brunetti1996microlocal, Hollands:2001nf, Hollands:2001fb, Hollands:2004yh}.
The backreaction of a quantum field on the curvature is taken into account by the semiclassical Einstein equation (SCE for shorts)
\begin{equation}
  \label{SCE}
  G_{a b} + \Lambda g_{ab} = 8\pi G \expvalom{T_{ab}},
\end{equation}
where $G_{ab}$ is the Einstein tensor, $\Lambda$ the cosmological constant, $g_{ab}$ the spacetime metric, $G$ the Newton constant, $\expvalom{T_{ab}}$ the expectation value of the quantum stress-energy tensor in a suitable state $\omega$ and we set $c = \hbar = 1$.
We observe that, in the algebraic approach, the requirements given by Wald \cite{Wald1977back, Wald:1978ce, Wald:1978pj} for any normal ordering prescription necessary to give a meaningful stress-energy tensor are satisfied \cite{Hollands:2001nf, Hollands:2004yh} and, if the chosen state $\omega$ is sufficiently regular, we also obtain finite expectation values \cite{brunetti1996microlocal}.

The quantum matter we shall consider in this paper is described by a real linear field whose classical equation of motion is 
\begin{equation}\label{eq:matter}
  -\square \phi + m^2\phi + \xi R\phi = 0,
\end{equation}
where $\square$ is the d'Alembert operator of the background metric, $m$ is the mass and $\xi$ describes the coupling to the scalar curvature $R$.
It is difficult to directly solve \eqref{SCE} for arbitrary values of the coupling constant $\xi$ because some contributions involving derivatives of the coefficients of the metric up to the fourth order are present in the expectation values of the stress-energy tensor.
This peculiar feature makes the semiclassical equation very different from its classical counterpart, which contains only second order derivatives of the metric.
Furthermore, these terms with higher order derivatives cannot be completely reabsorbed in a choice of the renormalization freedom present in the construction of local Wick polynomials (such as the stress-energy tensor) of the theory \cite{Hollands:2001nf, Hollands:2004yh}.
In fact, a careful analysis of the expectation value of the stress-energy tensor reveals that the term with the highest derivative appear in some non-local contributions (see the contribution to $\langle\wick{\phi^2}\rangle_\omega$ given in \eqref{T-Vk} of Proposition \ref{prop:decomposition}).
Actually, even if the normal ordered stress-energy tensor is a local quantum observable, its expectation values in a suitable state $\omega$ may involve some non-localities.
In particular, to fix the state one has to prescribe it on each point of a Cauchy surface, hence non-local contributions may arise in this way.
For this reasons, equation \eqref{SCE} cannot be written in normal form and thus a direct analysis of solution to \eqref{SCE} is problematic.

The problems with these higher order derivatives can be avoided only in special cases like the case of the massless conformally coupled scalar fields \cite{Wald:1978ce} and the case of massive conformally coupled fields \cite{Pinamonti:2010is, Pinamonti:2013wya}.
In these cases, the existence of a local and global solution can be obtained directly from \eqref{SCE}.
More recently, the problem of the existence of a solution of the semiclassical Einstein equation in cosmological spacetimes in the case of generic coupling has been addressed in \cite{Gottschalk:2018kqt}.
In that work, the semiclassical equation has been written as a dynamical system for the germ (the set of moments) of the finite part of the two-point function evaluated at coinciding points.
This dynamical system admits unique solutions when the chosen matter two-point function has good analytic properties like the two-point function of an equilibrium state.
The price to pay is that the dimension of that dynamical system is infinite and furthermore, for the case of a generic state, it is not clear if all these moments correspond to those obtained for a meaningful quantum state.
Another recent study of the initial value problem associated to semiclassical equations can be found in \cite{Juarez-Aubry:2019jon}.
Furthermore, numerical analyses have been performed in the past by Anderson in the study of effects of particle creation in the early universe \cite{Anderson:1983nq, Anderson:1984jf, Anderson:1985cw, Anderson:1985ds}.
A more recent numerical analysis of the semiclassical problem in cosmology performed in \cite{Nicolai} on the basis of the theoretical work in \cite{Gottschalk:2018kqt} obtained solutions which do not show an unphysical blowup.

In this paper we shall follow the approach presented in \cite{Pinamonti:2013wya} to prove the existence and uniqueness of local solutions of the full semiclassical Einstein equation in the case of a cosmological background without writing the system as an infinite dimensional dynamical system. Before discussing the details of the methods we shall use we recall some facts about the application in cosmology. In particular,we recall that some models of inflation are based on the analysis of the semiclassical version of the Einstein equations.
This is the case for the Starobinski model \cite{Starobinsky:1980te, Kofman:1985aw} where the higher derivative terms drive the expansion close to the Big Bang.
Solutions of this model have physical meaning in the regime where $R_{abcd}R^{abcd} \ll m_P^4$, here $m_P = \sqrt{\hbar c/ G}$ is the Planck mass, i.e., when quantum gravity effect can be assumed to be negligible.
It is furthermore claimed that this approximation holds only when the fluctuations inside the quantum stress-energy tensor are small \cite{Kuo:1993if}.
In this perspective, the validity of the semiclassical regime has been reformulated more recently in the framework of the so-called stochastic semiclassical gravity (or stochastic gravity), where the fluctuations of the stress-energy tensor are viewed as a stochastic source for the semiclassical equations \cite{Hu:1994dx, Hu:2008rga, Roura:1999qr, Roura:2007jj} (see also \cite{Pinamonti:2013zba}).
Moreover, a semiclassical analysis may help to study the formation of structures and galaxies.
It is often argued that these processes at large scales arise from small density perturbations at the early stages of the Universe, which are produced by the quantum fluctuations of a scalar field (the inflaton, the Higgs field for instance)
\cite{Starobinsky1982, Mukhanov:1981xt, PhysRevLett.49.1110, Hawking1982}.
For a summary of the discussions about the applications of semiclassical gravity we refer to \cite{Ford:2005qz}, furthermore, for some recent analyses of the $\Lambda$CDM-model in semiclassical gravity see \cite{Hack:2013uyu} and \cite{Matsui:2019tlf}.

As discussed above, in this paper we shall show that solutions of \eqref{SCE} exist for short finite intervals of time.
The main steps of the construction we shall present are the following.
We fix our attention to flat cosmological spacetimes, see \eqref{FLRW-metric} for the precise form of the metric. Since these spacetime are conformally flat, we shall use the conformal time \eqref{conformal-time} to describe the time evolution.
These spacetimes posses a single dynamical degree of freedom which is the scale factor $a$, and we use as dynamical equations the conservation of the stress-energy tensor and the trace of the semiclassical Einstein equation.
This system of equations is equivalent to the first Friedmann equation up to a constraint on initial conditions fixed at some initial time.
Hence, the system of equations we have to solve to determine the scale factor $a(\tau)$ for the conformal time $\tau$ contained in some interval $[\tau_0,\tau_1]$ is
\begin{subequations}
  \label{SCE-trace-intro}
  \begin{align}[left=\empheqlbrace]
    & \nabla_a \expvalom{T^a{}_b} = 0, \label{eq:SCE-conservation} \\
    & {-R} + 4\Lambda = 8\pi G \expvalom{T}, \label{eq:SCE-trace} \\
    & G_{00}(\tau_0) - a^2 \Lambda = 8\pi G \expvalom{T_{00}}(\tau_0), \label{eq:SCE-constraint}
  \end{align}
\end{subequations}
with the evolution of the state $\omega$ determined by the Klein--Gordon equation \eqref{eq:matter}. We prove that this system of equations can be solved once suitable initial conditions for $a$ and for the quantum state $\omega$ are fixed at $\tau=\tau_0$. In particular, since $\expvalom{T}$ contains fourth order derivatives of the scale factor, the initial conditions for $a$ fixes its derivative up to the third order. In order to fulfil the constraint at initial time described by \eqref{eq:SCE-constraint}, we do not put restrictions on the initial values of the scale factor but we look at this constraint as a limitation on the possible states for the quantum matter.

The state that we use needs to be homogeneous and isotropic, furthermore, for simplicity, we shall restrict our attention to the case of pure quasifree states. Notice that, in this paper, the state is used to obtain expectation values of the stress-energy tensor and of the Wick square. In general, since the field we are considering is linear, only the one-point function and the two-point function enter in the evaluation of these expectation values. The contribution of the one-point function can be understood as the contribution of the classical part of the field. The contribution of the two-point function can be analyzed as in this paper. For this reason, the request of being quasifree could be easily dropped. The request of being pure could also be dropped admitting two-point functions which are convex combinations of two-point functions of pure states. We recall in section \ref{se:pointsplitting} that it is possible choose a renormalization prescription for $T$ which ensures that the expectation values is conserved \cite{Hollands:2004yh}. Hence \eqref{eq:SCE-conservation} is always fulfilled by definition. Firstly, we prove in Proposition \ref{prop:Friedman-constraint} that the state for the quantum matter can be chosen to be regular enough to give finite expectation values of the stress-energy tensor and, furthermore, it can be chosen in such a way to solve the constraint mentioned above \eqref{eq:SCE-constraint} for every scale factor compatible with the chosen initial conditions. Secondly, we analyze the trace of the semiclassical Einstein equation \eqref{eq:SCE-trace} as the system of equations given in Proposition \ref{pr:SCE-wave}.
\begin{equation}
  \label{SCE-wave-intro}
    \begin{cases}
      (-\square + M_c ) F = S, \\
      \expvalom{\phi^2} - c_\xi R = F,
    \end{cases}
\end{equation}
where $c_\xi$ and $M_c$ are suitable constants, $S$ is a function of $a,a'$ and $a^{(2)}$ and where $\langle\wick{\phi^2}\rangle_\omega$ is the expectation of the normal ordered Wick square in the state $\omega$.

With a partial integration of this system of equation, namely solving the first equation in \eqref{SCE-wave-intro} for $F$ as discussed in Theorem \ref{theo:sol-SCE-wave}, we reduce the problem to the analysis of the single equation \eqref{SCE-state} or \eqref{eq:time-derivative-phi2}, which is equivalent to the second equation in \eqref{SCE-wave-intro}. This latter equation has the form 
\begin{equation}\label{eq:semiclassical-trace-T}
\partial_\tau\langle\wick{\phi^2}\rangle_\omega = \mathcal{S},
\end{equation} 
where $\langle\wick{\phi^2}\rangle_\omega$ is the expectation of the normal ordered Wick square in the state $\omega$ and where $\mathcal{S}$ is some source term depending on the curvature and on the various initial conditions for the scale factor and for the state (only third order derivatives of $a$ enters $\mathcal{S}$).
We then identify the term with the highest derivative which appears in the expectation value $\langle\wick{\phi^2}\rangle_\omega$ through the application of an unbounded linear operator (retarded) $\operatorT$, see Proposition \ref{prop:SCE-V}.
More precisely
\[
\partial_\tau\langle\wick{\phi^2}\rangle_\omega = \operatorT_{\tau_0}[f] + \mathcal{R},
\]
where $f =  \left(  m^2 a^2  + \at \xi - \frac{1}{6} \ct  R a^2 \right) '$ is the time derivative of the potential $V$ given in \eqref{potential} with respect to conformal time (this derivative is denoted by $'$). Thus, $f$ depends of $a$ and its derivative up to the third order, while $\mathcal{R}$ depends on the chosen state, on $a$ and its derivative up to the third order. Furthermore
\[
\operatorT_{\tau_0}[f](\tau)= -\frac{1}{8\pi^2}\int_{\tau_0}^{\tau} {f'}(\eta) \log(\tau - \eta) \dif \eta.
\]
This contribution does not depend on the state so it is not an artefact of the initial condition for the state.
We observe in Proposition \ref{prop:T} that this is the source of the loss of derivatives.
Actually, on small intervals of time $\operatorT$ satisfies the following inequality:
$\|\operatorT_{\tau_0}[f]\|_\infty \leq C \|f'\|_\infty$ for a fixed constant $C$. However, it is not continuous with respect to the uniform norm, so to control $\|\operatorT_{\tau_0}[f]\|_\infty$ we need fourth order derivatives of $a$.
However, we find an inversion formula for this operator in Proposition \ref{prop:T-inverse} and we show that we can reconstruct $f$ from $h=\operatorT_{\tau_0}[f]$ as $f=f(\tau_0) + \operatorT_{\tau_0}^{-1}[h]$. Furthermore, the inverse operator $\operatorT_{\tau_0}^{-1}$ appearing in the inversion formula \eqref{eq:inversion-formula} is more regular than $\operatorT_{\tau_0}$ and in particular we prove in equation \eqref{eq:T-continuity} of Proposition \ref{prop:T-inverse}
that it happens to be continuous with respect to the uniform norm. Hence, no loss of derivatives is introduced applying this inversion formula to \eqref{eq:semiclassical-trace-T} and the equation we get is 
\[
f = f(\tau_0) +  \operatorT_{\tau_0}^{-1}[\mathcal{S} - \mathcal{R}].
\]
Finally, in \eqref{SCE-fixed-point}, we rewrite this equation as a fixed point equation
\[
X'= \mathcal{C}[X'] 
\]
for $X'$, where $X$ given in \eqref{X} is related to the scale factor by $X = \frac{1}{6}a^2R = a''/a$. Furthermore, that fixed point equation is constructed with a map $\mathcal{C}$ introduced in Lemma \ref{lem:cT} which acts on a suitable compact subset of $C[\tau_0,\tau_1]$. Notice that the initial conditions fix the derivative of the scale factor $a$ up to the third order at $\tau_0$ and hence we can associate to every of $X'$ a unique scale factor $a$ integrating the equation $a'' = X a$ once these initial conditions are known. In Proposition \ref{prop:contraction} we prove that the map $\mathcal{C}$ is a contraction map if the time interval $[\tau_0,\tau_1]$ on which it is analyzed is sufficiently small.
In Theorem \ref{theo:contractio}, existence and uniqueness is then obtained by applying the Banach fixed point theorem.

The structure of the paper is the following: in the next section we give a brief description of the classical cosmological scenario and we present the basic tools necessary to discuss the quantization of the real quantum scalar field on flat Friedmann-Lema\^itre-Robertson-Walker (FLRW) spacetimes.
In particular we discuss the Hadamard point-splitting procedure to regularize composite fields like the stress-energy tensor, the energy density or $\phi^2$.
In section \ref{se:states} we discuss the regularity conditions on the quantum state which are necessary to obtain finite expectation values of the energy density and of $\phi^2$ and we give an estimate for the expectation value of $\phi^2$ and its first time derivative in these states.
In section \ref{se:SEE-integration} we discuss the semiclassical Einstein equation as the system of equations formed by the traced semiclassical equation and a constraint which needs to be fixed at initial time.
We show that the initial constraint can always be fulfilled and we partially integrate the trace equation.
The problem of finding solutions of the SCE is thus reduced to the problem of finding solutions of a single equation.
In section \ref{se:SEE-fixed-point} we discuss the properties of that equation.
In particular, we isolate the contribution with the highest derivative and we show how to write this equation as a fixed point equation inverting a certain unbounded operator.
We finally discuss the existence and uniqueness of the solutions of that equation.
The last section contains an outlook on possible future developments.
Some technical propositions and lemmas are collected in the appendix.

\section{Quantum field theory on cosmological spacetimes}\label{section:qft}

\subsection{Friedmann-Lema\^itre-Robertson-Walker spacetime}

According to the cosmological principle and recent observations, our universe is homogeneous and isotropic at large scales and it is essentially spatially flat, hence it can be accurately described by a flat FLRW spacetime $(\mathcal{M},g)$ where $\cM = I_t \times \Sigma$, $I_t \subset \mathbb{R}$ is an interval of time and $\Sigma$ is a three dimensional Euclidean space.
The metric is
\begin{equation}
  \label{FLRW-metric}
  g = -\dif t \otimes \dif t+a(t)^{2} \sum_{i=1}^3\dif x^{i} \otimes \dif x^{i},
\end{equation}
where the Euclidean coordinates $\vec{x} = (x^1,x^2,x^3)$ are the comoving coordinates of an isotropic observer while $t$ denotes \emph{cosmological time}.
The strictly positive function $a(t)$ is the \emph{scale factor} which is the unique degree of freedom of the spacetime.
It describes the ``history'' of our universe and is determined by solving the Einstein equations.

Every flat FLRW spacetime is conformally flat as can be seen writing the metric \eqref{FLRW-metric} with respect to \emph{conformal time}
\begin{equation}
  \label{conformal-time}
  \tau \doteq \tau_0 + \int_{t_0}^{t} \frac{\dif\eta}{a(\eta)}.
\end{equation}
In local conformal coordinates $(\tau,\vec{x})$, the metric is
\begin{equation}
  \label{flat-FLRW}
  g = a(\tau)^2 \Bigl( -\dif\tau \otimes \dif\tau + \sum_{i=1}^3 \dif x^i \otimes \dif x^i \Bigr),
\end{equation}
viz., FLRW spacetimes are conformally related to the Minkowski spacetime by a conformal transformation whose conformal factor is $a(\tau)$.
In the following we shall consider the scale factor $a(\tau)$ as a function of the conformal time.
Derivatives with respect to conformal time will be denoted by primes and derivatives with respect to cosmological times by dots, i.e., for the first derivatives of a time-dependent function $f$ we write $f'$ and $\dot{f}$, respectively.

\begin{remark}
  \label{rem:scale-factor}
  As already pointed out for instance in \cite{Agullo:2014ica}, the semiclassical Einstein equation for $\xi \neq \frac16$ involves always up to four time derivatives of the scale factor $a(\tau)$, due to the mass dimension of the stress-energy tensor as composite operator, which is equal to four.
  Thus, in the case of strong solutions, $a(\tau)$ has to be at least a $C^4$ function.
\end{remark}

The request of having an homogeneous and isotropic solution imposes constraints on the stress-energy tensor $T_{ab}$ which sources the Einstein equation.
Both in comoving and in conformal coordinates it must have the form $T_{a}{}^b = \operatorname{diag}(-\varrho, p, p, p)$, where $\varrho$ is the matter's \emph{energy density} and $p$ its \emph{pressure}.
Since the stress-energy tensor is covariantly conserved, i.e., $\nabla_a {T^a}_b = 0$, the Einstein equation reduces to the first Friedmann equation \begin{equation*} H^2 = \frac{8 \pi G}{3} \varrho + \frac{\Lambda}{3}, \end{equation*} where $H \doteq \frac{\dif}{\dif t} \log{a}$ is the Hubble function.
Analogously, if the stress-energy tensor is conserved, the dynamics of $a$ is determined by the traced Einstein equation,
\begin{equation}
  \label{trace-equation}
  - R + 4 \Lambda = 8\pi G T,
\end{equation}
together with an initial condition which corresponds to the validity of the first Friedmann equation at an initial time $\tau = \tau_0$, i.e.,
\begin{equation}
  \label{00-equation}
  H(\tau_0)^2 = \frac{8\pi G}{3}\varrho(\tau_0) + \frac{\Lambda}{3}.
\end{equation}
We shall adopt this second set of equations in the semiclassical analysis.

\subsection{Scalar quantum field}

In this work we consider a very simple kind of quantum matter: a real linear quantum massive Klein-Gordon field coupled to curvature with a generic coupling, the corresponding classical equation of motion is \eqref{eq:matter}.
In order to deal with the semiclassical Einstein equation \eqref{SCE} we have to analyze the expectation value of the stress-energy tensor of this system, hence, we have to discuss the quantization of the system and we have to select a quantum state.

The quantization of this scalar field can be performed constructing the algebra of observables generated by the quantum field $\phi$ \cite{Haag:1992hx,LibroAQFT} implementing the canonical commutation relations (CCR).
In particular, on every smooth globally hyperbolic spacetime $(\cM,g)$ one can construct the CCR algebra of quantum fields $\cA(\cM,g)$ as the $*$-algebra generated by $\{\phi(f),f\in\mathcal{D}(\mathcal{M})\}$ which is the set of linear fields smeared with compactly supported smooth functions satisfying the following relations \[ \phi({P} f)=0, \qquad \phi(f)^{*}=\phi(\bar{f}), \qquad[\phi(f), \phi(h)]=\mathrm{i} \Delta(f, h), \] where $f,h$ are compactly supported smooth functions, namely elements of $\mathcal{D}(\cM)=C^\infty_0(\cM)$, $P$ is the Klein-Gordon operator and $\Delta=\Delta_R-\Delta_A$ is the \emph{causal propagator} on $(\cM,g)$ defined as the difference of the unique retarded and advanced fundamental solution of $P\phi = 0$.

Thanks to the conformal flatness of the metric, the Klein-Gordon operator on FLRW spacetimes can be written in conformal time $\tau$ as
\begin{equation}
  \label{KG-operator}
  P= -\square + \xi R + m^{2} = \frac{1}{a^{3}}\left(\partial_\tau^{2} - \vec{\nabla}^2 + a^2 \left(\xi-\frac{1}{6}\right) R + a^{2} m^{2}\right) a,
\end{equation}
where $\square \doteq g^{ab}\nabla_a\nabla_b$ and $\vec{\nabla}^2 \doteq \sum_i \partial_i^2$ denotes the spatial Laplacian operator with respect to the (comoving) spatial coordinates, $m$ is the mass and $\xi$ is the coupling constant to the scalar curvature.

In the algebraic language, a quantum state $\omega$ is a positive, normalized, linear functionals over $\mathcal{A}(\mathcal{M},g)$.
Since $\mathcal{A}(\mathcal{M},g)$ is generated by linear fields, the state is determined once the $n$-point functions $\omega_n\in\mathcal{D}'(\mathcal{M}^n)$ are given \[ \omega_n(f_1,\dots, f_n) \doteq \omega(\phi(f_1)\dots \phi(f_n)).
\]
On curved spacetime there is no preferred vacuum to be used as reference state.
Here we shall choose a state which is at least quasifree (Gaussian) and pure.
The $n$-point functions of quasifree states are completely determined once the two-point function is given.
On $(\mathcal{M},g)$ the two-point function of a pure state which is homogeneous and isotropic is of the form \cite{luders1990local, Hack:2015zwa}
\begin{equation}
  \label{Gaussian-state}
  \omega_{2}(x,y)=\lim _{\varepsilon \rightarrow 0^{+}} \frac{1}{(2 \pi)^{3}} \int_{\mathbb{R}^{3}} \frac{\overline{\zeta}_k\left(\tau_x\right)}{a\left(\tau_x\right)} \frac{\zeta_k\left(\tau_y\right)}{a\left(\tau_y\right)} \mathrm{e}^{\mathrm{i} \vec{k} \cdot(\vec{x}-\vec{y})} \mathrm{e}^{-\varepsilon k} \dif \vec{k},
\end{equation}
where $k \doteq |\vec{k}|$ and where the \emph{temporal modes} $\zeta_k$ fulfil the equation
\begin{equation}
  \label{mode-eq}
  \zeta''_k(\tau) + \Omega_k^2(\tau) \zeta_k(\tau)=0, \qquad \Omega_k^2(\tau) \doteq k^2 + a^{2}m^{2} + \at \xi-\frac{1}{6} \ct Ra^2,
\end{equation}
and satisfy the normalization condition
\begin{equation}
  \label{normalization}
  \zeta'_k \overline{\zeta}_k-\zeta_k \overline{\zeta}'_k=\mathrm{i}.
\end{equation}

In this paper we shall consider only cases where $\Omega_k^2(\tau) > 0$ for every $k$.
Once $m>0$ and $\xi$ are fixed, this will be done assuming suitable initial conditions for $R$ and restricting the time interval accordingly. Equation \eqref{mode-eq} and the condition \eqref{normalization} do not fix the modes uniquely and, as already said, on a generic FLRW spacetime there is no preferred choice. However, here we are interested in computing expectation values of the stress-energy tensor. For this reason, the modes we shall select need to give a state which is regular enough to have a finite expectation value of the normal ordered stress-energy tensor.

\subsection{Point splitting regularization}\label{se:pointsplitting}

Local fields like $\phi^2$ or $T_{ab}$, necessary for the analysis of the semiclassical equation, are not elements of $\cA(\cM,g)$, furthermore, their expectation values on generic states diverge. That is because these fields are products of fields at the same point and $\omega_n$ are distributions which have singularities in the coinciding point limits.
This problem is usually overcome considering normal ordered fields. Actually, physically relevant states are those for which normal ordered fields have finite expectation values and are called Hadamard states. Furthermore, all such states have a universal divergence \cite{KAY199149,Wald1977back}.
Hence, the idea beyond the normal ordering prescription is to subtract these universal divergences before taking the coinciding point limit.
The fields obtained in this way are covariant because only local geometry enters in the construction of the subtraction \cite{brunetti1996microlocal, Hollands:2001nf}.
In this procedure there is a freedom which has been classified by Hollands and Wald in \cite{Hollands:2004yh} and for every normal ordered fields it amounts to fixing a finite number of renormalization constants.

To be more precise, we recall that the singularity is universal for Hadamard states.
Moreover, thanks to the work of Radzikowski \cite{radzikowski1996micro}, a quasifree state is Hadamard if and only if its two-point function fulfils the microlocal spectrum condition, see also \cite{brunetti1996microlocal}.
The two-point function of a Hadamard state in a convex geodesic neighbourhood $\mathcal{O}$ is always given by 
\[ 
\omega_2(x,y) \doteq \cH(x,y) + w(x,y) = \lim _{\epsilon \rightarrow 0^{+}} \cH_{\varepsilon}(x,y)+ w(x,y),
\] 
where the limit is taken in the distributional sense and where
\begin{equation}
  \label{Hadamard-singularity}
  \cH_{\varepsilon}(x,y) \doteq \frac{u(x,y)}{\sigma_{\varepsilon}}+v(x,y) \log \left(\frac{\sigma_{\varepsilon}}{\lambda^{2}}\right)
\end{equation}
is the Hadamard singularity, with $\sigma_{\varepsilon}(x,y)= \sigma(x,y) + \mathrm{i}\epsilon( t(x)-t(y)) $.
Here $\sigma(x,y)$ is one half of the geodesic distance between $x$ and $y$ taken with sign and $t$ is any time function.
Furthermore, $\lambda$ is a length scale.
The so-called Hadamard coefficients $u(x,y), v(x,y) = \sum_n v_n(x,y) \sigma(x,y)^n$ and $w(x,y)$ are smooth functions on $\mathcal{O} \times \mathcal{O}$: $u$ and $v$ are real-valued bi-scalars fixed by the metric and the equation of motion $P\phi = 0$, while $w(x,y)$ characterizes the state and must be chosen in such a way that $\omega_{2}$ is a bi-solution of the Klein-Gordon equation.
On Hadamard states, normal ordered fields can be obtained by a point-splitting regularization \cite{Brunetti:1999jn, Hollands:2001nf} which consists of subtracting the divergences contained in $\cH$ before computing the coinciding point limits.

The classical form of the stress-energy tensor is
\begin{align*}
  T_{ab} & = \nabla_a \phi \nabla_b \phi - \frac{1}{2} g_{ab} (\nabla_c\phi \nabla^c \phi +m^2 \phi^2) + \xi \left(G_{ab} \phi^2 -\nabla_a\nabla_b\phi^2 + g_{ab} \nabla_c\nabla^c \phi^2 \right).
\end{align*}
Notice that in terms like $\nabla_a\phi^2$ and $\nabla_a\nabla_b\phi^2$ the normal ordering prescription is implemented before applying the covariant derivatives.
Furthermore, since $\nabla_a \phi \nabla_b \phi = \nabla_a \nabla_b \phi^2 - \phi \nabla_a \nabla_b \phi$, we just need to discuss the normal ordering of
$\Psi_{ab} \doteq \phi \nabla_a \nabla_c \phi$ and $\Psi \doteq \phi^2$, see e.g. \cite{Hollands:2001nf, Hollands:2004yh, moretti2003comments}.
Their expectation values are thus obtained as
\begin{equation}
  \label{phi2-Hadamard}
  \begin{aligned}
    \expval{\wick{\phi^2}(x)}_\omega & = \lim_{y \rightarrow x} \left( \omega_2(y,x) - \mathcal{H}(y,x)\right)  = \lim_{y \rightarrow x} w(y,x),                                                           \\
    \expval{\wick{\phi \nabla_a \nabla_b \phi}(x)}_\omega & = \lim_{y \rightarrow x} \nabla^{(x)}_a \nabla^{(x)}_b\left( \omega_2(y,x) - \mathcal{H}(y,x)\right)= \lim_{y \rightarrow x}  \nabla^{(x)}_a \nabla^{(x)}_b w(y,x).
  \end{aligned}
\end{equation}
However, the normal ordering prescription of defining local Wick polynomials fixes the fields only up to certain combinations of local curvature terms and the mass.
Imposing some fundamental constraints like locality, scaling behaviour, covariance, this freedom can be classified \cite{Hollands:2001nf} and it reduces to the freedom of fixing a finite number of renormalization constants.
In the case of $\phi^2(x)$, one can define a new equivalent Wick monomial \[ \wick{\tilde{\phi}^2}(x) = \wick{\phi^2}(x) + \tilde{\alpha}_1 R(x) + \tilde{\alpha}_2 m^2 \] for arbitrary real renormalization constants $\tilde{\alpha}_1,\tilde{\alpha}_2$.
The freedom in the construction of $\phi \nabla_a \nabla_b \phi$ is further constrained by the requirement that $\nabla_a{T^a}_{b}=0$.
The resulting renormalization freedom of the stress-energy tensor is thus \[ \wick{\tilde{T}_{ab}}(x) = \wick{T_{ab}}(x) + \tilde{\beta}_1 m^4 g_{ab} +  \tilde{\beta}_2 m^2 G_{ab} + \tilde{\beta}_3' I_{ab} + \tilde{\beta}_4' J_{ab}, \] where $\tilde{\beta}_i$ are renormalization constants, where the tensors $I$ and $J$ are obtained as functional derivatives of $\sqrt{g} R^2$ and $\sqrt{g} R^{ab}R_{ab}$ and contain up to fourth order derivatives of the metric \cite{Wald:1978pj}.
For conformally flat spacetimes like FLRW, $I_{ab} = 3 J_{ab}$ and, furthermore, their traces are both proportional to $\square R$, namely ${I^a}_a = 3{J^a}_a = 6\square R$.

However, since $\mathcal{H}_\epsilon$ is only a bisolution of the equation of motion up to a smooth term, imposing the constraint $\nabla_a{T^a}_{b}=0$ results in an anomalous contribution to the trace of $T$ known as trace anomaly \cite{Hollands:2004yh, moretti2003comments}.
In particular
\begin{equation}
  \label{traceT}
  \expvalom{T} = \left(3\left(\xi - \frac{1}{6}\right) \square - m^{2}\right)\expvalom{\phi^2} + \frac{1}{4 \pi^{2}}\left[v_{1}\right] +4c_1 m^{4} - c_2 m^{2}
  R - c_3 \square R,
\end{equation}
where $c_1,c_2$ and $c_3$ are the renormalization constants, $[v_1]$ is the coinciding point limit of the Hadamard coefficient $v_1$ and it is of the form
\begin{equation*}
  \begin{split}[v_1] &= \frac{m^4}{8} + \frac{(6\xi - 1)m^2R}{24} + \frac{(6\xi - 1)^2 R^2}{288} + \frac{(5\xi -1)\square R}{120} + \frac{R_{abcd}
    R^{abcd} - R_{ab}R^{ab}}{720} \\ &= \frac{m^4}{8} + \frac{(6\xi - 1)m^2}{4} \frac{{a''}}{a^3} + \frac{(6\xi - 1)^2}{8}\frac{{a''}^2}{a^6} + \frac{(5\xi -1)\square R}{120} + \frac{1}{60} \left( \frac{{a'}^4}{a^8} - \frac{{a''}{a'}^2}{a^7} \right).
  \end{split}
\end{equation*}

We finally observe that, for our purposes, we just need to implement the point splitting regularization for $\Psi$ and $\Psi_{ab}$.
Hence, we might subtract from the two-point function the truncated Hadamard parametrix at order $n$ with $n=1$, where 
\[ 
\cH_{n}(x,y) \doteq \lim _{\epsilon \rightarrow 0^{+}}  \frac{u(x,y)}{\sigma_{\varepsilon}}+ \sum_{k=0}^{n}v_k(x,y)\sigma^k(x,y) \log \left(\frac{\sigma_{\varepsilon}}{\lambda^{2}}\right), 
\] 
see e.g. \cite{Christensen:1976vb, Christensen:1978yd}.
If we regularize $\omega_2$ with the truncated Hadamard parametrix at order $1$, we get that $\omega_2-\cH_{1}$ is only a $C^2$ function \cite{KAY199149} (see also \cite{Gottschalk:2018kqt}).
We furthermore observe that the regularization procedure for $T_{ab}$ can be extended beyond Hadamard states to states whose finite part $w$ is only $C^2$.

We finally notice that, in the expectation values of the trace of the stress-energy tensor, derivatives of the coefficients of the metric up to the fourth order appear.

\section{Sufficiently regular states}\label{se:states}

In this paper we consider a state $\omega$ which is quasifree, pure, homogenous and isotropic and it is described by the two-point function $\omega_2$ of the form given in \eqref{Gaussian-state}.
Such a state is thus completely characterized by the initial conditions for the modes $\zeta_k$ used to define \eqref{Gaussian-state}. We notice that two-point functions \eqref{Gaussian-state} constructed with modes which differ by a global phase which is constant in time coincide.
Furthermore, the normalization condition \eqref{normalization} is a constraint on the initial initial conditions for the modes.
Actually, if we decompose $\zeta_k = \rho \mathrm{e}^{i\theta}$ with $\rho$ and $\theta$ real, the normalization condition \eqref{normalization} implies that ${\theta}'=(2\rho^2)^{-1}$.
Hence, $\theta$ can be obtained from $\rho$ because the modes needs to be fixed up to a global phase and we may assume $\theta(\tau_0)=0$.
So, the initial conditions at $\tau_0$ for the modes, and hence for the state, are fully specified by $\rho(\tau_0)$ and $\rho'(\tau_0)$.
We shall equivalently characterize the state by $\Phi(k)$ and $E(k)$, two real functions of $k$, and a sign $s\in\{-1,+1\}$ which fix the \emph{initial conditions} of the modes $\zeta_k$
\begin{equation}
  \label{eq:init-cond-state}
  \Phi(k) = |{\zeta_k}(\tau_0)|^2 = |\rho(\tau_0)|^2 , \qquad
  E(k)=|{{\zeta}'_k}(\tau_0)|^2 =  |{\rho}'(\tau_0)|^2 + \frac{1}{4\Phi(k)}, \qquad \text{sign} \left(\text{Re} \left({\zeta'_k}(\tau_0)\right)\right) =s.
\end{equation}
The functions $\Phi(k)$ and $E(k)$ must satisfy the following inequalities 
\[ 
E(k) \geq \frac{1}{4\Phi(k)} \geq 0 
\] 
necessary to give origin to meaningful initial conditions for the modes
\[ 
\rho_k(\tau_0) = \sqrt{\Phi(k)}, \qquad \rho_k'{(\tau_0)} = s\sqrt{ E(k)-\frac{1}{4\Phi(k)}}.
\]

To keep some generality, we shall not make any particular choice for the functions $\Phi, E$ and $s$.
However, we shall always assume that these functions are chosen in such a way that the corresponding state is sufficiently regular, namely that it gives finite expectation values for observable like $\wick{\phi^2}$ and the energy density $\wick{\varrho} \doteq \wick{T_{00}} = -\wick{{T_0}^0}$ involved in the semiclassical equations 
\eqref{SCE-trace-intro}, or in \eqref{trace-equation} and \eqref{00-equation}.
Furthermore, in order to have a well-defined semiclassical Einstein equation, these functions may depend on the derivative of $a$ up to the third order.
These conditions are met by adiabatic states of fourth order \cite{Parker:1969au, luders1990local, Junker:2001gx} or by the instantaneous vacuum states considered by Agullo et all.
in \cite{Agullo:2014ica}.
More precisely, in view of \eqref{traceT}, the relevant observables that we need to control are the Wick square $\wick{\phi^2}$ and the energy density $\wick{\varrho}$.
Their expectation values can be obtained following the analyses performed in \cite{Schlemmer:2010lte, Eltzner:2011dyn, Degner:2013ota, Hack:2013uyu, Siemssen:2015owa} and in the state \eqref{Gaussian-state} they take the form
\begin{equation}
  \label{eq:rho}
  \begin{aligned}
    \langle\wick{\phi^2}\rangle_{\omega} & = \frac{1}{(2\pi)^3 a^2} \int_{\mathbb{R}^{3}}\left( |\zeta_k|^2 -C^\cH_{\phi^2}(\tau, k) \right) \dif \vec{k} +
    \frac{w(\tau)^2}{8\pi^2a^2}\log(\frac{w(\tau_0)}{a(\tau)})
    -
    \frac{w(\tau_0)^2}{16\pi^2 a^2}+{\alpha}_1 m^2 + {\alpha}_2 R,       \\
    \langle\wick{\varrho}\rangle_{\omega}
    & =
    \frac{1}{(2\pi)^3 a^4} \int_{\mathbb{R}^{3}} \left( \frac{|{\zeta}'_k|^2}{2}+\left(k^2 +a^2m^2
    -\left(6\xi-1\right)a^2H^2
    \right) \frac{|\zeta_k|^2}{2}
    +a H \left(6\xi-1\right) 2\text{Re}(\overline{\zeta}_k\zeta'_k) \right. \\
   &  -C^\cH_{\varrho}(\tau, k) \bigg) \dif \vec{k}
     \quad
    -\frac{H^4}{960\pi^2}
    +\left(\xi-\frac{1}{6} \right)^2\frac{3 H^2R}{8\pi^2} +\tilde{\beta}_1 m^4 -\tilde{\beta}_2 m^2 {G_{0}}^0+(\tilde{\beta}_3-\frac{\tilde{\beta}_4}{3}) {I_{0}}^0,
  \end{aligned}
\end{equation}
where, $k=|\vec{k}|$ and, recalling \eqref{mode-eq}, $w(\tau) = \sqrt{\Omega(\tau)^2-k^2}= a\sqrt{m^2 + \left(\xi-\frac{1}{6}\right) R}$. Here, ${\alpha_i}$ and $\tilde{\beta}_i$ are (redefinitions of the) renormalization constants of the theory.
Furthermore, the functions $C^\cH_{\phi^2}(\tau, k)$ and $C^\cH_{\varrho}(\tau, k)$ are subtracted before the $k$-integration to implement the point splitting regularization mode-wise.
{After introducing the initial frequency $k_0^2 \doteq \Omega_k^2(\tau=\tau_{0})$ and the perturbative potential
\begin{equation}
    \label{potential}
    V(\tau) \doteq \Omega_k^2(\tau) - k^2_0 = m^2(a^2 - a_0^2) + \at \xi - \frac{1}{6} \ct (R a^2 - R_0 a_0^2),	\qquad 	\qquad a_0 = a(\tau_{0}), \quad R_0=R(\tau_{0}),
  \end{equation} we define}
\begin{equation}
  \label{eq:C-norm}
  \begin{aligned}
    C^\cH_{\phi^2}(\tau, k) & \doteq \frac{1}{2 k_{0}} - \frac{V(\tau)}{4k_0^3}, \\  C^\cH_{\varrho}(\tau, k) & \doteq \frac{k}{2} + \frac{a^2m^{2}-a^2H^{2}(6 \xi-1)}{4 k} - \frac{a^4 m^{4} + 12\left(\xi-\frac{1}{6}\right)m^{2}a^4H^2 + a^4\left(\xi - \frac{1}{6}\right)^2 2{I_{0}}^0(\tau)}{16 k (k^2+  \frac{a^2}{\lambda^2})},
  \end{aligned}
\end{equation}
where  $2{I_{0}}^0 = 216 H^{2} \dot{H}-36 \dot{H}^{2}+72 H \ddot{H}$ and it corresponds to the $00$-component of the local curvature tensor ${I_{ab}}$ which encompasses part of the renormalization freedom of $T_{ab}$. $\lambda$ is the length scale present in the Hadamard singularity. Notice that we are in the case where $\Omega_k(\tau_0)^2$ is strictly positive thanks to the choice of initial conditions for the spacetime we are considering.
\begin{definition}\label{def:sufficiently-regular}
  We say that a pure homogeneous and isotropic quasi-free state whose two-point function is constructed as in \eqref{Gaussian-state} with  modes $\zeta_k$ is \emph{sufficiently regular} if
  \begin{equation}
    \label{eq:init-regular1}
    |\zeta_k^2(\tau_0)| - C^\cH_{\phi^2}(\tau_0, k) \in L^1(k^2 \dif k),
    \quad
    \left.\frac{\dif}{\dif\tau} \left[|\zeta_k^2|(\tau) - C^\cH_{\phi^2}(\tau, k)  \right]\right|_{\tau=\tau_0} \in L^1(k^2 \dif k)
  \end{equation}
  and
  \begin{equation}
    \label{eq:init-regular2}
    \left.\at \frac{|\zeta'_k|^2}{2}+\left(k^2 +a^2m^2-\left(6\xi-1\right)a^2H^2\right) \frac{|\zeta_k|^2}{2}+a H \left(6\xi-1\right) 2\text{Re}(\overline\zeta_k\zeta'_k)-C^\cH_{\varrho}(\tau, k) \ct\right|_{\tau=\tau_0} \in L^1(k^2 \dif k).
  \end{equation}
\end{definition}
As we will see in Proposition \ref{prop:FV}, it is just sufficient to demand the regularity stated in Definition \ref{def:sufficiently-regular} at initial time in order to ensure the finiteness of the expectation values of $\wick{\phi^2}$ and $\wick{\varrho}$, namely the observables appearing the semiclassical Einstein equation for cosmological spacetimes.
\begin{remark}
 We observe that in order to check if a state given in \eqref{Gaussian-state} is sufficiently regular we need to have control on the derivatives of the scale factor up to the third order because no fourth order derivative of the metric appear in $C^\cH_{\varrho}(\tau_0, k)$, in $ C^\cH_{\phi^2}(\tau_0, k)$ and in $\partial_\tau C^\cH_{\phi^2}(\tau_0, k)$ and the same holds for the corresponding finite contributions. Furthermore, adiabatic states of fourth order are sufficiently regular in the sense of Definition \ref{def:sufficiently-regular} (see for instance \cite{PhysRevD.9.341, PhysRevD.36.2963}).
\end{remark}

\subsection{Expectation values of \texorpdfstring{$\phi^2$}{\textphi?} and its time derivative}

For later purposes we need to control the expectation values of $\phi^2$ in a quasifree state $\omega$ whose two-point function \eqref{Gaussian-state} is constructed with modes $\zeta_k$ which satisfy the initial conditions \eqref{eq:init-cond-state} chosen in such a way that point splitting regularization works, namely \eqref{eq:init-regular1} and \eqref{eq:init-regular2} hold.
In particular, we need to know how the state depends on the scale factor $a$ and on the initial conditions $a_0,a'_0,a''_0$ and $a^{(3)}_0$.
To control how $\langle\wick{\phi^2}\rangle_\omega$ evolves in time, we compare the modes $\zeta_k$ with some reference modes $\chi_k$ which are solutions of \eqref{mode-eq}, and which satisfy the normalization condition \eqref{normalization} and are characterized by the following initial values
\begin{equation}
  \label{eq:init-value}
  \chi_k\left(\tau_0\right) = \frac{1}{\sqrt{2 k_{0}}} \mathrm{e}^{\mathrm{i} k_{0} \tau_0}, \qquad
  {\chi}'_k\left(\tau_0\right) = \frac{\mathrm{i} k_{0}}{\sqrt{2 k_{0}}} \mathrm{e}^{\mathrm{i} k_{0} \tau_0},
\end{equation}
where we recall that $k_0 = \sqrt{\Omega_k(\tau_0)}$, with $\Omega_k$ given in \eqref{mode-eq}.
Furthermore, the parameter of the theory are fixed in such a way that $k_0$ is strictly positive for every $k$.
Notice that since the mass $m$ is strictly positive and $\xi$ is fixed, we have the room of making this choice by restricting the possible initial conditions $a_0$ and $a_0''$.
The quasifree state $\omega^c$ constructed with the modes $\chi$ is called \emph{conformal vacuum}, its two-point function is \[ \omega^c_{2}(x,y)=\lim _{\varepsilon \rightarrow 0^{+}} \frac{1}{(2 \pi)^{3}} \int_{\mathbb{R}^{3}} \frac{\overline{\chi}_k\left(\tau_x\right)}{a\left(\tau_x\right)} \frac{\chi_k\left(\tau_y\right)}{a\left(\tau_y\right)} \mathrm{e}^{\mathrm{i} \vec{k} \cdot(\vec{x}-\vec{y})} \mathrm{e}^{-\varepsilon k} \dif \vec{k}.
\]
Notice that, after subtracting the Hadamard singularity \eqref{Hadamard-singularity}, we obtain a function which is only continuous.
It is nevertheless useful to compute the expectation value of $\phi^2$ in this state and to compare it with the one in the state $\omega$.
To this end we observe that the modes $\chi_k$ can be constructed with a convergent Dyson series.
Actually, we have the following proposition taken from \cite{Pinamonti:2010is, Pinamonti:2013wya} 

\begin{proposition}
  \label{prop:modes}
  Consider the FLRW spacetime $(\cM,g)$ with $a\in C^2(\cM)$, constructed in such a way that $\Omega_k^2(\tau_0)$ in \eqref{mode-eq} is strictly positive, a solution $\chi_k$ of \eqref{mode-eq} which satisfies the initial conditions \eqref{eq:init-value}  can be obtained explicitly on $[\tau_0,\tau_1]$ as
  \begin{equation}
    \label{eq:series}
    \chi_k=\sum_{n \geq 0} \chi_k^{n},
  \end{equation}
  where $\chi_k^{n}$ for $\tau>\tau_0$ are obtained recursively.
  The recursive step is for $n>0$
  \begin{equation}
    \label{chi-n}
    \chi_k^n(\tau) = - \int_{\tau_0}^\tau \frac{\sin(k_0(\tau-\eta))}{k_0}
    V(\eta) \chi_k^{n-1}(\eta) \dif\eta, \qquad \chi^0_k(\tau) = \frac{1}{\sqrt{2 k_{0}}} \mathrm{e}^{\mathrm{i} k_{0} \tau},
  \end{equation}
  where $k_0=\Omega(\tau_0)$ and $V(\tau)$ is the perturbation potential \eqref{potential}.
  The following bound holds
  \begin{equation}
    \label{chi-n-estimate}
    |\chi_k^n| \leq \frac{1}{\sqrt{2 k_{0}} n !}\left(\frac{1}{k_{0}} \int_{\tau_{0}}^{\tau}|V(\eta)| \dif \eta\right)^{n} \leq \frac{1}{\sqrt{2k_0} n!} \frac{(\tau-\tau_0)^n}{k_0^{n}} \| V \|_\infty^n.
  \end{equation}
  Hence, the series \eqref{eq:series} converges absolutely and
  \begin{equation}
    \label{modes-estimate}
    |\chi_k(\tau)| \leq \frac{1}{\sqrt{2 k_{0}}} \exp( \frac{\| V\|_{1,[\tau_0,\tau]}}{k_0}  ),
    \qquad
    |\chi_k'(\tau)| \leq \left(\sqrt{\frac{k_{0}}{2}} + \frac{\| V'\|_{1,[\tau_0,\tau]}}{\sqrt{2} k_0^{3/2}}  \right)
    \exp(  \frac{2\| V\|_{1,[\tau_0,\tau]}}{k_0}  ),
  \end{equation}
  where the norm $\|\cdot\|_{1,[\tau_0,\tau]}$ is the ordinary $L^1$ norm on the interval $[\tau_0,\tau]$. Furthermore, 
    \begin{equation}
    \label{eq:chi-chi0}
    \begin{aligned}
    |(\chi_k-\chi^0_k)|&\leq \frac{\|V\|_{1,[\tau_0,\tau]}}{\sqrt{2}k_0^{3/2}} \exp \left(\frac{ \|V\|_{1,[\tau_0,\tau]}}{k_0}  \right), \qquad
	\\
    |(\chi_k-\chi^0_k)'| &\leq \left( \frac{\|V\|_{1,[\tau_0,\tau]}}{\sqrt{2k_0}} + \frac{\|V'\|_{1,[\tau_0,\tau]}}{\sqrt{2}{k_0}^{3/2}} \right) \exp \left(\frac{2\|V\|_{1,[\tau_0,\tau]}}{k_0}   \right).
    \end{aligned}
  \end{equation}
\end{proposition}
\begin{proof}
Equation \eqref{mode-eq} equipped with initial conditions  \eqref{eq:init-value} form a well posed Cauchy problem hence an unique $\chi$ solution exists. Furthermore, since \eqref{mode-eq} is of the form \eqref{eq:second-order-ode}, we may apply the results of Lemma \ref{le:second-order-ODE}. In particular, \eqref{eq:rewert-equation} implies that
\begin{equation}\label{eq:first-chi}
\chi_k(\tau) =   - \Delta_R^{k_0}*V\chi_k +\chi_k^0,  
\end{equation}
where $\chi_k^0 = \frac{1}{\sqrt{2k_0}}e^{\mathrm{i}k_0 \tau}$. Hence 
\begin{equation}\label{eq:rec}
(1 - \mathcal{R}) \chi_k =   \chi_k^0,  
\end{equation}
where the linear operator $\mathcal{R}$ is such that $\mathcal{R}\chi_k = - \Delta_R^{k_0}*V\chi_k$. Applying the inverse of $(1 - \mathcal{R})$ on both side of \eqref{eq:rec} we obtain $\chi_k$ in terms of $\chi_k^0$. 
Actually, 
\[
\chi_k = \sum_{n\geq0} \mathcal{R}^n \chi_k^0 =\sum_{n\geq 0} \chi_k^n
\]
and since $\chi_k^n = \mathcal{R}^n \chi_k^0$ we obtain \eqref{chi-n}. In particular, expanding $\mathcal{R}^n\chi_k^0$, we have
\[ 
\chi_k^n(\tau_{n+1})  =(-1)^n\int_{\tau_{0} \leq \tau_{1} \leq \cdots \leq \tau_{n+1}} \prod_{j=1}^{n} \left( \frac{\sin(k_0(\tau_{j+1}-\tau_j))}{k_0} V(\tau_j)   \right)   \chi_k^0(\tau_1) \dif\tau_1 \dots \dif\tau_n 
\] 
from which we obtain \eqref{chi-n-estimate}. Absolute convergences of the series $\sum_{n\geq 0} \chi_k^n$ to $\chi_k$ together with its first and second derivatives can now be obtained analyzing the explicit form of $\mathcal{R}$ and using \eqref{eq:first-chi}.
The first estimates in \eqref{modes-estimate} and in \eqref{eq:chi-chi0} can be obtained by an application of Gr\"onwall lemma as in Lemma \ref{le:second-order-ODE} from the inequalities
\[
|\chi_k(\tau)|\leq \frac{1}{\sqrt{2k_0}} + \int_{\tau_0}^\tau \frac{|V(\eta)|}{{k_0}}|\chi_k(\eta)| \dif \eta, 
\]

\[
|(\chi_k - \chi^0_k)(\tau) |\leq  
\frac{1}{\sqrt{2}}\int_{\tau_0}^\tau  \frac{|V|}{k_0^{3/2}} \dif \eta + \int_{\tau_0}^\tau \frac{|V(\eta)|}{{k_0}}|(\chi_k-\chi_k^0)(\eta)| \dif \eta 
\]
which are obtained directly from \eqref{eq:first-chi}. 
The second estimates in \eqref{modes-estimate} and in \eqref{eq:chi-chi0} descend from the first estimates and applying Gr\"onwall lemma  to the inequalities
\[
|\chi'_k(\tau)|\leq \sqrt{\frac{k_0}{2}} +  \int_{\tau_0}^\tau \frac{|V'(\eta)|}{{k_0}}|\chi_k(\eta)|\dif \eta   + \int_{\tau_0}^\tau \frac{|V(\eta)|}{{k_0}}|\chi'_k(\eta)| \dif \eta, 
\]

\[
|\chi_k-\chi_k^0|' =  \int_{\tau_0}^\tau  \frac{|V|}{k_0} |\chi_k-\chi_k^0|'\dif\eta + \int_{\tau_0}^\tau\frac{V'}{k_0} |\chi_k|  \dif \eta,
+
\int_{\tau_0}^\tau\frac{V}{\sqrt{2k_0}} |\chi_k|  \dif \eta.
\]
These inequalities are obtained directly from 
\[
\chi_k' =  -\Delta_R^{k_0}* V \chi_k' - \Delta_R^{k_0}* V'\chi_k  + {\chi^0_k}',
\qquad
(\chi_k-\chi_k^0)' =  -\Delta_R^{k_0}* V (\chi_k-\chi_k^0)' - \Delta_R^{k_0}* V' \chi_k   - \Delta_R^{k_0}* V (\chi^0_k)'
\]
which is the first derivative of \eqref{eq:first-chi}.
\end{proof}

We now decompose the expectation value of $\phi^2$ and of its time derivative in the state $\omega$ which is regular, namely it is quasifree and its two-point function is constructed as in \eqref{Gaussian-state} with modes $\zeta_k$ satisfying \eqref{eq:init-regular1} \eqref{eq:init-regular2}.
We have
\begin{equation}
  \label{eq:Dphi2-decomposition}
  \langle\wick{\phi^2}\rangle_{\omega} =  \frac{Q_s}{a^2}+ \frac{Q_c}{a^2}+\frac{Q_0}{a^2},\qquad \partial_\tau \left( a^2\langle\wick{\phi^2}\rangle_{\omega}\right) =  {Q_s^d}+ {Q_c^d}+{Q_0^d},
\end{equation}
where the state dependent contribution is contained in the following
\begin{align*}{Q_s}
   & \doteq
  a^2\langle\wick{\phi^2}\rangle_{\omega}  -a^2
  \langle\wick{\phi^2}\rangle_{\omega^c}
  =
  \frac{1}{(2\pi)^3} \int_{\mathbb{R}^{3}}\left( |\zeta_k|^2 - |\chi_k|^2\right) \dif \vec{k},  \\
  {Q_s^d} & \doteq \frac{1}{(2\pi)^3} \int_{\mathbb{R}^{3}}\left( \partial_\tau |\zeta_k|^2 - \partial_\tau |\chi_k|^2 +\frac{{V}'(\tau_0)}{4k_0^3} \cos(2k_0(\tau-\tau_0))\right) \dif \vec{k}
\end{align*}
and the subtraction of $C^\mathcal{H}_{\phi^2}$ in \eqref{eq:C-norm} taken before the $k$-integration is visible in the following contributions
\begin{align*}{Q_c}   & \doteq \lim_{\epsilon \rightarrow 0^{+}} \frac{1}{(2\pi)^3} \int_{\mathbb{R}^{3}}\left[ |\chi_k|^2 - \left(\frac{1}{2 k_{0}} - \frac{V(\tau)}{4 k_0^3}\right)\right] \mathrm{e}^{-\epsilon k} \dif \vec{k},
  \\
  {Q^d_c} & \doteq \lim_{\epsilon \rightarrow 0^{+}} \frac{1}{(2\pi)^3} \int_{\mathbb{R}^{3}}\left[ \partial_\tau|\chi_k|^2 + \left( \frac{V'(\tau)}{4 k_0^3}
    -\frac{{V}'(\tau_0)}{4k_0^3} \cos(2k_0(\tau-\tau_0))
    \right)\right] \mathrm{e}^{-\epsilon k} \dif \vec{k}.
\end{align*}
\begin{remark}
  \label{re:subtraction}
  We observe that the subtraction considered in $Q^d_c$ differs from $\partial_\tau C^\mathcal{H}_{\phi^2}$ by a contribution which is compensated in $Q^d_s$.
  This extra subtraction is necessary because the conformal vacuum $\omega^c$, namely the Gaussian state constructed with the modes $\chi_k$, is not regular enough to give finite time derivatives of $\langle \wick{\phi^2} \rangle_{\omega^c}$.
\end{remark}
Finally, the other two contributions $Q_0$ and $Q_0^d$ are obtained from \eqref{eq:Dphi2-decomposition} as the reminder.
Both are functions of $a$ and its derivatives and contain the finite reminder of the $C^\mathcal{H}_{\phi^2}$ subtraction discussed in \eqref{eq:rho}:
\begin{align}
  {Q_0}   & \doteq a^2\langle\wick{\phi^2}\rangle_{\omega^c}  - {Q_c}
  =
  \frac{w(\tau)^2}{8\pi^2}\log(\frac{w(\tau_0)}{a(\tau)})
  -
  \frac{w(\tau_0)^2}{16\pi^2}
  +    \alpha_1 m^2 a^2+  \alpha_2 a^2R
  \notag,
  \\
  \label{eq:Qd0}
  {Q^d_0} & \doteq
  \frac{\partial_\tau w(\tau)^2}{8\pi^2}\log(\frac{w(\tau_0)}{a(\tau)})
  - \frac{aH w(\tau)^2}{8\pi^2}
  +  \alpha_1 m^2 \partial_\tau (a^2)  +  \alpha_2 \partial_\tau (a^2R),
\end{align}
where $\alpha_1$ and $\alpha_2$ are renormalization constants and where we recall that $w(\tau)=a\sqrt{m^2 + \left(\xi-\frac{1}{6}\right) R}$.

We shall now analyze these contributions separately. In particular, we need to know how they depend on $V$ through the scale factor $a$. Hence in the next, we shall bound them and their Gateaux differential to get the Lipschitz continuity of these quantities. To this end we recall some definitions and some facts in the following Remark. 
\begin{remark}\label{re:functionalderivative}
  Consider a functional $F:\mathcal{D}\to\mathbb{R}$ where $\mathcal{D}$ is some Banach space.
  The functional derivative or Gateaux differential of $F$ at $V\in\mathcal{D}$ in the direction $W\in\mathcal{D}$ is defined as the following limit
  \[
    \delta F[V,W] \doteq \lim_{\epsilon\to 0} \frac{ F[V+\epsilon W] - F[V]}{\epsilon},
  \]
  where the limit $\epsilon\to0$ is taken with respect to the norm topology of $\mathcal{D}$.
  If the functional derivative at $V$ exists for every direction $W\in\mathcal{D}$ and if $\delta F[V,W]$ is linear and continuous in $W$ we say that $F$ is Gateaux differentiable in $V$ and in this case $\delta F[V,W]$ is called Gateaux derivative of $F$ in $V$.
  To get Lipschitz continuity, we observe that, if the functional derivative $\delta F[V,W]$ depends continuously on $W$ uniformly in $V$, namely if
  \[
    |\delta F[V,W]| \leq C \|W\|
  \]
  for some constant $C$ which does not depend on $V$, we have
  \[
    F[V_1]-F[V_2]
    = \int_0^1 \frac{\dif}{\dif\epsilon} F[V_2+\epsilon(V_1-V_2)] \dif\epsilon
    = \int_0^1 \delta F[\epsilon V_1+(1-\epsilon)V_2,V_1-V_2] \dif\epsilon
  \]
  hence the Lipschitz continuity is obtained:
  \[
    |F[V_1]-F[V_2]| \leq C \| V_1-V_2\|.
  \]
  Furthermore, later, we shall consider composition of functionals, in that case we shall evaluate Lipschitz continuity in the following way.
  Notice that if $F$ depends on $V$ through a function $A[V]$ with a functional dependance on $V$, namely $F[V] = \tilde{F}[A[V]]$ and if both $\tilde{F}$ and $A$ are Gateaux differentiable and if they are bounded
  \[
    |\delta\tilde{F}[A,B]| \leq C_1 \|B\|,
    \quad
    \|\delta A[V,W]\| \leq C_2 \|W\|,
  \]
  with $C_1$ and $C_2$ which do not depend on $A$ and $V$, we have that
  \[
    \delta F[V,\delta V] = \delta \tilde{F}[A,\delta A[V,\delta V]]
  \]
  and in this case
  \[
    |\delta F[V,\delta V]| \leq C_1 C_2 \|W\|.
  \]
  hence in this case
  \[
    |F[V_1]-F[V_2]| \leq C_1C_2 \| V_1-V_2\|,
  \]
  thus obtaining the desired Lipschitz continuity.
\end{remark}

\begin{proposition}
  \label{prop:decomposition}
  Consider a cosmological spacetime and an interval of time $[\tau_0,\tau_1]$ over which $\Omega_k^2$ given in \eqref{mode-eq} is positive.
  Consider the following non-linear operators acting on $C^2$-functions which vanish at $\tau_0$, namely on $D^2\doteq \{V\in C^2[\tau_0,\tau_1]\; |\; V(\tau_0)=0\}$
  \begin{align*}
    Q_c[V](\tau) & =\lim_{\epsilon \rightarrow 0^{+}} \frac{1}{(2\pi)^3} \int_{\mathbb{R}^{3}}\left[ |\chi_k|^2 - \left(\frac{1}{2 k_{0}} - \frac{V(\tau)}{4 k_0^3}\right)\right] \mathrm{e}^{-\epsilon k} \dif \vec{k}, \\ {Q^d_c}[V](\tau) & = \lim_{\epsilon \rightarrow 0^{+}} \frac{1}{(2\pi)^3} \int_{\mathbb{R}^{3}}\left[ \partial_\tau|\chi_k|^2 + \left( \frac{V'(\tau)}{4 k_0^3} -\frac{{V}'(\tau_0)}{4k_0^3} \cos(2k_0(\tau-\tau_0)) \right)\right] \mathrm{e}^{-\epsilon k} \dif \vec{k},
  \end{align*}
  where $\chi_k$ is the solution of \eqref{mode-eq} with initial data \eqref{eq:init-value} and thus it implicitly depends on $V$.
  Consider also the following operator
  \begin{equation}
    \label{T-Vk}
    \begin{aligned}
      \operatorT_{\tau_0}[f] & \doteq -\frac{1}{8\pi^2} \int_{\tau_0}^{\tau} {f}'(\eta) \log(\tau - \eta) \dif\eta,\qquad f\in D^2.
    \end{aligned}
  \end{equation}
  It holds that the functionals \[ Q_f[V] \doteq Q_c[V] - \operatorT_{\tau_0}[V],  \qquad Q_f^d[V] \doteq Q^d_c[V] - \operatorT_{\tau_0}[V']  \] admit the Gateaux differential at $V$ or $V'$. Furthermore, $Q_f$ is continuous with respect to the uniform norm on the interval $[\tau_0,\tau]$ and the same holds for its first functional derivative, hence $Q_f$ can be extended to continuous functions which vanish at $\tau_0$, namely to $D^0\doteq\{V\in C[\tau_0,\tau_1] \mid V(\tau_0)=0\}$.
  If $V$ is contained in $B_\delta(0)$, a ball of radius $\delta$ centred at $0$ in $C[\tau_0,\tau_1]$, then \[ \|Q_f[V]\|_\infty \leq C_{\delta}\| V \|_\infty, \qquad \|\delta Q_f[V,W]\|_\infty \leq C_\delta'\| W' \|_\infty, \qquad V\in B_\delta(0) \cap D^0 \subset C[\tau_0,\tau]. \]
  Similarly, $Q^d_f$ is continuous with respect to the uniform norm of the derivative on the interval $[\tau_0,\tau_1]$ and it can then be extended to $D\doteq \{ V\in C^{1}[\tau_0,\tau_1] \mid V(\tau_0)=0\}$. For $V\in D$ and if $V'$ is contained in $B_\delta(0) \subset C[\tau_0,\tau_1]$ then \[ \|Q^d_f[V]\|_\infty \leq C_{\delta}\| V' \|_\infty, \qquad \|\delta Q^d_f[V,W]\|_\infty \leq C_\delta'\| W' \|_\infty, \qquad V\in D,\; V'\in B_\delta(0), \] where the constants $C_{\delta}, C_{\delta}'$ depend smoothly on $\delta$ and are bounded uniformly in time for $\tau-\tau_0 <\epsilon$ for some $\epsilon>0$.
\end{proposition}

\begin{proof}
  We recall the results of Proposition \ref{prop:modes}, hence $\chi= \sum_n \chi^n$. We then observe that $Q_c[V]$ and $Q^d_c[V]$ can be decomponsed in contributions which are homogenous in $V$ of various degrees
  \[
    Q_c[V]= \sum_{n\geq 0}{L_n}[V],\qquad Q^d_c[V]= \sum_{n\geq 0}{L^d_n}[V].
  \]
 
  We observed that both zeroth order contributions vanish because of the form of $\chi_k^0$ given in \eqref{chi-n}. Furthermore,
  \begin{align*}
    L_1 & = \lim_{\epsilon \rightarrow 0^{+}} \frac{1}{(2\pi)^3} \int_{\mathbb{R}^{3}}\left[  (\overline{\chi}^1_k \chi^0_k+ {\chi}^1_k \overline{\chi}^0_k) + \frac{V(\tau)}{4 k_0^3}\right] \mathrm{e}^{-\epsilon k} \dif \vec{k}, \\ L^d_1 & = \lim_{\epsilon \rightarrow 0^{+}} \frac{1}{(2\pi)^3} \int_{\mathbb{R}^{3}}\left[  \partial_\tau(\overline{\chi}^1_k \chi^0_k+ {\chi}^1_k \overline{\chi}^0_k) + \left( \frac{V'(\tau)}{4 k_0^3} -\frac{{V}'(\tau_0)}{4k_0^3} \cos(2k_0(\tau-\tau_0)) \right) \right] \mathrm{e}^{-\epsilon k} \dif \vec{k},
  \end{align*}
  while for $n \geq 2$ the limit $\epsilon\to0$ can be taken before the $k$-integration, hence \[ L_n[V] =  \sum_{l=0}^{n}  \frac{1}{(2\pi)^3}\int_{\mathbb{R}^3} \overline{\chi_k^{n-l}} \chi_k^{l} \dif\vec{k}, \qquad L^d_n[V]  = \partial_\tau L_n[V].
  \]
  To study the form of $L_1$ and of $L^d_1$, we recall the definition of $\chi^1$, integrating by parts and using the condition $V(\tau_0)=0$, we obtain
  \begin{align*}
    \overline{\chi}_k^{1}\chi_k^{0} +\overline{\chi}_k^0\chi_k^{1} & = -\frac{V(\tau)}{4k_0^3}+ \frac{1}{4 k_{0}^3} \int_{\tau_0}^{\tau} \cos \left(2 k_0(\tau-\eta)\right) {V'}(\eta) \dif\eta, \\
    \partial_\tau \at \overline{\chi}_k^{1}\chi_k^{0} +\overline{\chi}_k^0\chi_k^{1} \ct & = -\frac{V'(\tau)}{4k_0^3}+\frac{V'(\tau_0) \cos(2k_0(\tau-\tau_0))}{4k_0^3}+ \frac{1}{4 k_{0}^3} \int_{\tau_0}^{\tau} \cos \left(2 k_0(\tau-\eta)\right) {V''}(\eta) \dif\eta.
  \end{align*}
  We discuss in details the construction of $L_1^d[V]$, $L_1$ can then be obtained in a similar way,
  \begin{align*}
    L^d_1[V](\tau)    = & \lim_{\epsilon\to0}\frac{1}{ 8\pi^2} \int_{\tau_0}^\tau \dif\eta \, {V''}(\eta) \int_{w_0}^\infty \dif k_0 \, \aq \cos\left(2k_0 (\tau-\eta)\right)  \frac{1}{k_0} - \cos\left(2k_0 (\tau-\eta)\right)  \frac{k_0-k}{k_0^2} \cq \mathrm{e}^{-2 \epsilon k_0},
  \end{align*}
  where $w_0 \doteq \sqrt{k_0^2-k^2} = \sqrt{a^{2}(\tau_0)m^{2}+\left(\xi-\frac{1}{6}\right) a^2R(\tau_0)}$ is the $k$-independent part of $\Omega_k(\tau_0)$. The $k_0$-integration in the second contribution gives \[ f_1(w_0 (\tau-\eta)) \doteq \int_{w_0}^\infty \cos (2k_0 (\tau-\eta))   \frac{k_0-k}{k_0^2} \dif k_0.
  \]
  Hence $f_1\in C^1(\mathbb{R})$ thus, on compact intervals, both $f_1,\partial_\tau f_1$ are bounded because $\frac{(k_0-k)}{k_0^2} = \frac{w_0^2}{k_0^2(k_0+k)}$.
  The $k_0$-integration in the first contribution, can be performed and in the limit $\epsilon\to0$ it gives \[ \mathrm{Ci}(2w_0(\tau-\eta)) = -\int_{w_0}^{\infty} \frac{\cos(2k_0(\tau-\eta))}{k_{0}}   \dif k_0.
  \]
  Here, $\mathrm{Ci}(z)$ is the cosine integral function which can be expanded as \cite{abramowitz1965handbook}
  \[
    \mathrm{Ci}(z) =  \gamma + \log(z)  + \int_0^z \frac{\cos(t)-1}{t}\dif t,
  \]
  where $\gamma$ is the Euler-Mascheroni constant. Then,
  \begin{align*}
    L^d_1[V](\tau) & = -\frac{1}{ 8 \pi^2} \int_{\tau_0}^\tau \dif\eta {V''}(\eta) \mathrm{Ci}(2 w_0 (\tau -\eta))   - \frac{1}{ 8 \pi^2} \int_{\tau_0}^{\tau} \dif\eta {V''} (\eta) f_1(w_0(\tau-\eta)).
  \end{align*}
  Integrating by parts and recalling the definition of $\operatorT$,  we get
  \begin{align*}
    L^d_1[V]-\operatorT_{\tau_0}[V'] = & -\frac{1}{8\pi^2} (\gamma + \log(2w_0) + f_3(0)) V'(\tau) +\frac{1}{8\pi^2} (\gamma + \log(2w_0) + f_3(w_0(\tau-\tau_0))) V'(\tau_0) \\  & - \frac{w_0}{8\pi^2}\int_{\tau_0}^\tau \dif\eta V'(\eta) f_3'(w_0 (\tau-\eta)),
  \end{align*}
  where $f_3(x)\doteq f_1(x)+f_2(x) \in C^{1}(\mathbb{R})$ because the function \[ f_2\left(z\right)  \doteq \mathrm{Ci}(2z) - \gamma - \log(2z) \] is of class $C^{1}(\mathbb{R})$ and it is thus bounded on finite interval of times.
  We thus have that, on the interval $[\tau_0,\tau_1]$, \[ \|L^d_1[V]-\operatorT_{\tau_0}[V']\|_\infty \leq C \|V'\|_\infty,  \] where the constant $C$ depends continuously on $\tau_1$ and vanishes in the limit $\tau_1\to \tau_0$.
  Since both $L^d_1$ and $\operatorT_{\tau_0}$ are linear in $V$ this proves also the Gateaux differentiability and its corresponding bounds.
  Similar results holds also for $L_1[V]-\operatorT_{\tau_0}[V]$.
  
  To analyze the order $n=2$ we observe that
  \begin{align*}
    \left(\chi_k^{0} \overline{\chi}_k^{2}+\chi_k^{1} \overline{\chi}_k^{1}+\chi_k^{2} \overline{\chi}_k^{0}\right)(\tau) & = \frac{1}{k_0^3} \int_{\tau_0}^\tau \dif\eta \, V(\eta)  		\sin(k_0(\tau-\eta)) \int_{\tau_0}^\eta \dif\xi \, V(\xi) \sin(k_0(\tau+\eta-2\xi )) \\  & = \frac{1}{2k_0^3} \int_{\tau_0}^\tau \dif\eta \, V(\eta) \int_{\tau_0}^\eta \dif\xi V(\xi)  \at \cos(2k_0(\eta-\xi)) - \cos(2k_0(\tau-\xi))\ct,
  \end{align*}
  hence, for the second order, we have
  \begin{align*}
    L_2[V](\tau) = & \frac{1}{8\pi^3} \int_{\mathbb{R}^3} \frac{\dif \vec{k}}{2k_0^3} \int_{\tau_0}^\tau \dif\eta \, V(\eta) \int_{\tau_0}^\eta \dif \xi V(\xi)  \at\cos(2k_0(\eta-\xi)) - \cos(2k_0(\tau-\xi))\ct, \\ L^d_2[V](\tau) = & \frac{1}{8\pi^3} \int_{\mathbb{R}^3} \frac{\dif \vec{k}}{2k_0^3} \int_{\tau_0}^\tau \dif\eta \, V(\eta)^2 \cos(2k_0(\tau-\eta)) - \frac{1}{8\pi^3} \int_{\mathbb{R}^3} \frac{\dif \vec{k}}{2k_0^3} \int_{\tau_0}^\tau \dif\eta \, V(\eta) \int_{\tau_0}^\eta \dif \xi V'(\xi) \cos(2k_0(\tau-\xi)),
  \end{align*}
  where we used the initial condition $V(\tau_0)=0$.
  Notice that the integral in $k_0$ can be computed just as in the linear case.
  However, now logarithmic divergences in $\xi-\eta$ and $\xi-\tau$ for $L_2[V]$ and in $\xi-\tau$ and $\eta-\tau$ for $L_2^d[V]$ are absolutely integrable.
  These logarithmic divergences can be identified before taking the $\epsilon$ to $0$ limit switching the order of $\eta$ and $\vec{k}$ integration.
  Hence, on the interval $[\tau_0,\tau_1]$, we have \[ \|L_2[V]\|_\infty \leq C \|V\|_\infty^2, \qquad |\delta L_2[V,W]| \leq C \|V\|_\infty\|W\|_\infty \] and, since $V\in D$,  \[ \|L^d_2[V]\|_\infty \leq C \|V\|_\infty\|V'\|_\infty, \qquad |\delta L^d_2[V,W]| \leq C \|V'\|_\infty\|W'\|_\infty, \] where $C$ is a suitable constant which depends continuously on $\tau_1$ and vanishes in the limit $\tau_1\to\tau_0$.

  Furthermore for $n>2$ it holds that \[ L_c[V] = \sum_{n\geq 3} L_n{[V]} =\lim_{\epsilon \rightarrow 0^{+}} \sum_{n\geq 3} \frac{1}{(2\pi)^3} \int_{\mathbb{R}^{3}} \sum_{l=0}^n \left(\overline{\chi_k^{n-l}}\chi_k^{l} \right) \mathrm{e}^{-\epsilon k} \dif \vec{k} \] and $L_c^d[V] = \partial_\tau L_c[V]$.
  From the inequality \eqref{chi-n-estimate},
  \[
    |L_c[V]| \leq \sum_{n\geq 3} \frac{1}{(2\pi)^3} \int_{\mathbb{R}^{3}}
    \frac{1}{2k_0} \frac{2^n(\tau-\tau_0)^n}{k_0^{n}} \frac{1}{n!} \|V\|_\infty^n \dif \vec{k} \leq C \frac{2^2}{w_0^4} \|V\|_\infty^3 \exp\left(2\frac{(\tau-\tau_0)}{w_0} \|V\|_\infty\right),
  \]
  where we used the fact that $\sum_{l=0}^n\frac{1}{l!} \frac{1}{(n-l)!} = \frac{2^n}{n!} $.
  To bound $L_c^d[V]$, we rewrite \eqref{chi-n} in the following compact form \[ \chi_k^n = \mathcal{R}(\chi_k^{n-1}) = \mathcal{R}^n ( \chi_k^0), \] where the operator $\mathcal{R}$ acts on a function $f$ as $\mathcal{R}(f)=-\Delta_R (Vf)$ and $\Delta_R$ is the retarded operator defined in \eqref{chi-n}. In view of the fact that $V(\tau_0)=0$, we have
  \[ 
  {\chi_k^n}' =  -\Delta_R V' \chi_k^{n-1} -\Delta_R V {\chi_k^{n-1}}'.
  \]
  Hence, using recursively the previous identity, we get an expression which depends linearly on $V'$.
  This expression has the form
  \[
    {\chi_k^n}' =  \sum_{j = 0}^{n-1}  \underbrace{ \mathcal{R}\circ \dots \circ \mathcal{R}}_{j} \circ \tilde{\mathcal{R}}  \circ \underbrace{ \mathcal{R}\circ \dots \circ \mathcal{R}}_{n-1-j} (\chi_k^0)
    + \underbrace{ \mathcal{R}\circ \dots \circ \mathcal{R}}_{n}({\chi_k^0}'),
  \]
  where $\tilde{\mathcal{R}}(f)=-\Delta_R (V'f)$
  and can be written in a compact form in the following way with the help of a functional derivative which transforms $\mathcal{R}$ to $\tilde{\mathcal{R}}$
  \[
    {\chi_k^n}' =  \int \dif\eta V'(\eta)\frac{\delta \chi_k^n}{\delta V(\eta)}
    + \mathcal{R}^n({\chi_k^0}').
  \]
  Furthermore $\overline{\chi_k^0}'=  - \mathrm{i} k_0 \overline{\chi_k^0}$ and ${\chi_k^0}'=  \mathrm{i} k_0 {\chi_k^0}$, hence
  \begin{equation}
    \label{eq:step-func}
    (\overline{\chi}_k^l {\chi}_k^{n-l})' =  \int \dif\eta V'(\eta)\frac{\delta}{\delta V(\eta)} \overline{\chi}_k^l {\chi}_k^{n-l}.
  \end{equation}
  In the estimate we have derived above for $L_c[V]$ we have bounded $|\sin(k_0(\eta))/k_0| \leq 1/k_0$ without touching $V$, hence, to obtain an estimate for $L_c^d[V]$ we may just apply the operator with the functional derivative introduced on the right hand side of \eqref{eq:step-func} to derive an estimate for $L^d_c[V]$. On an interval $[\tau_0,\tau_1]$ we have \[ \|L_c[V]\|_\infty \leq C \|V\|_\infty^3 \exp(C \|V\|_\infty), \qquad \|L^d_c[V]\|_\infty \leq C \|V\|_\infty^2 \|V'\|_\infty \exp(C \|V\|_\infty), \] where the constant $C$ depends continuously on $\tau_1$ and vanishes in the limit of $\tau_1\to\tau_0$.
  A similar analysis permits to get analogous estimates for the first functional derivatives \[ \|\delta L_c[V,W]\|_\infty \leq C \exp(C \|V\|_\infty) \|W\|_\infty, \] where again the constant $C$ depends continuously on $\tau_1$ and vanishes in the limit of $\tau_1\to\tau_0$.
  The statements of the proposition, namely the Gateaux differentiability in $D$ with respect to the uniform norm and its bounds, can be obtained combining the obtained estimate for $L_1$, $L_2$,  $L_c$ and assuming $V\in B_\delta(0)$ or combining the estimates for $L^d_1$, $L^d_2$, $L^d_c$ and assuming $V'\in B_\delta(0)$.
\end{proof}

\begin{proposition}
  \label{prop:FV}
  Consider a FLRW spacetime $(\cM,g)$, whose scale factor is $a(\tau) \in C^{3}[\tau_0,\tau]$ with $a(\tau)>0$, and the quasifree state $\omega$ given in \eqref{Gaussian-state} with respect to modes $\zeta_k$ whose initial conditions \eqref{eq:init-cond-state} satisfy \eqref{eq:init-regular1} and \eqref{eq:init-regular2}.
  We furthermore assume that on the interval $[\tau_0,\tau_1]$, $\Omega_k^2$ given in \eqref{mode-eq} is strictly positive.
  The non-linear operator \[ Q_s[V] = a^2 \left(\expval{\wick{\phi^2}}_{\omega} - \expval{\wick{\phi^2}}_{\omega^c} \right)=\frac{1}{(2\pi)^3} \int_{\mathbb{R}^{3}}\left( |\zeta_k|^2 - |\chi_k|^2\right) \dif \vec{k} \] is Gateaux differentiable at $V$ in $D^0=\{V\in C[\tau_0,\tau_1]\mid V(\tau_0)=0\}$, where the initial conditions for the modes $\chi_k$ are given in \eqref{eq:init-value}.
  Similarly, the non-linear operator \[ Q^d_s[V]  = \frac{1}{(2\pi)^3}  \int_{\mathbb{R}^{3}}\left( \partial_\tau |\zeta_k|^2 - \partial_\tau |\chi_k|^2 +\frac{{V}'(\tau_0)}{4k_0^3} \cos(2k_0(\tau-\tau_0))\right) \dif \vec{k} \] is Gateaux differentiable at $V'\in C[\tau_0,\tau_1]$.
\end{proposition}
\begin{proof}
    We observe that the modes $\zeta_k$ can be written as a linear combination of the modes $\chi_k$, namely, $\zeta_k = A \chi + B\overline{\chi_k}$, where $A=A(k)$ and $B=B(k)$ are the Bogoliubov coefficients and they can depend on $k=|\vec{k}|$ but not on $\tau$.
  Furthermore since both $\chi_k$ and $\zeta_k$ satisfy the normalization condition \eqref{normalization} we have that $|A|^2 - |B|^2=1$. Hence,
  \begin{equation}
    \label{eq:QSV}
    Q_s[V]  = \frac{1}{(2\pi)^3} \int_{\mathbb{R}^{3}}\left( 2 |B|^2 |\chi_k|^2 + A\overline{B} \chi_k\chi_k +    \overline{A}B \overline{\chi_k\chi_k} \right) \dif \vec{k}.
  \end{equation}

  We shall now control how this expression depends on $V$. In particular, we get the properties and the form of $A$ and $B$ from the requirements  \eqref{eq:init-regular1} and \eqref{eq:init-regular2} and then at a later time we control the evolution of $Q_s$ from the known evolution of the modes $\chi_k$ discussed in Proposition \ref{prop:modes}.
  To analyze the form of $A$ and $B$ for large values of $k$ we notice that the initial conditions for $\zeta_k$ are chosen in such a way that \eqref{eq:init-regular1} and \eqref{eq:init-regular2} holds hence we have that $|\zeta_k|^2-C^\mathcal{H}_{\phi^2}$ is in $L^1(\mathbb{R}^3,\dif\vec{k})$ at $\tau_0$.
  The same holds for $\partial_\tau |\zeta_k|^2-\partial_\tau C^\mathcal{H}_{\phi^2}$ at $\tau_0$, hence, also $\partial_\tau|\zeta_k|^2/ \Omega_k$ is absolutely $\vec{k}$-integrable because $1/\Omega_k$ is a bounded function at times larger or equal than $\tau_0$ and because $\partial_\tau C^\mathcal{H}_{\phi^2}/\Omega_k$ is absolutely $\vec{k}$-integrable at any time larger or equal than $\tau_0$ as can be seen from the definition of $C^\mathcal{H}_{\phi^2}$ and of $\Omega_k$.
  Furthermore, from \eqref{eq:init-regular2} we have that it must exists a function $C^\mathcal{H}_E$ which depends on $\tau$ and $k$ which makes $|\zeta_k'|^2-C^\mathcal{H}_E$ absolutely $\vec{k}$-integrable at $\tau_0$.
  Hence at $\tau_0$ it is absolutely integrable also the quantity $(|\zeta_k'|^2-C^\mathcal{H}_E)/\Omega_k^2$.
  Notice that, $C^\mathcal{H}_E$ can be obtained arguing as for $C^\mathcal{H}_\varrho$, $C^\mathcal{H}_{\phi^2}$ and $\partial_\tau C^\mathcal{H}_{\phi^2}$ and has the form
  \[
    C^\cH_{E}(\tau, k) =
    \frac{k_0}{2} + \frac{V(\tau)}{4 k_0} + O\left(\frac{1}{k_0^3}\right).
  \]
  Consider now the functions
  \[
    f_1= \frac{|\zeta'_k|^2}{\Omega_k^2}-|\zeta_k|^2, \qquad f_2=  |\zeta_k|^2-|\chi_k|^2;
  \]
  both are absolutely $\vec{k}$-integrable at $\tau_0$ because $\frac{C^\mathcal{H}_E}{\Omega_k^2}-C^\mathcal{H}_{\phi^2}$ is in $L^1(\mathbb{R},\dif\vec{k})$ and the same holds for
  $|\chi_k|^2-C^\mathcal{H}_{\phi^2}$ and because we already know that $(|\zeta_k'|^2-C^\mathcal{H}_E)/\Omega_k^2$ and $|\zeta_k|^2-C^\mathcal{H}_{\phi^2}$ are absolutely $\vec{k}$-integrable at $\tau_0$.
  Notice that $\frac{\partial_\tau f_2}{\Omega_k}$ is also absolutely $\vec{k}$-integrable at $\tau_0$ because $\partial_\tau C^\mathcal{H}_{\phi^2} /\Omega_k$ is in $L^1(\mathbb{R}^3,\dif\vec{k})$.
  If we now evaluate $f_1$ and $\partial_\tau f_2/\Omega_k$ at time $\tau_0$ we get 
  \[ 
  	f_1(\tau_0) = -\frac{2\text{Re}\left(\overline{A}B\mathrm{e}^{\mathrm{i}2k_0\tau_0}\right)}{k_0}, \qquad \frac{\partial_\tau f_2(\tau_0)}{\Omega_k(\tau_0)} = \frac{2\mathrm{i}\text{Im}\left(\overline{A}B\mathrm{e}^{\mathrm{i}2k_0\tau_0}\right)}{k_0}
  \]
  which must be both elements of $L^1(\mathbb{R},\dif\vec{k})$. Hence $|AB|/k_0$ is also absolutely integrable.
  The same holds also for $\frac{|B|^2}{k_0}$ because it can be obtained adding $f_1(\tau_0)/2$ to the absolute integrable function 
  \[ 
  	f_2(\tau_0)= \frac{|B|^2}{k_0} +\frac{\text{Re}\left(\overline{A}B\mathrm{e}^{\mathrm{i}2k_0\tau_0}\right)}{k_0}.
  \]
  These estimates involve only initial conditions for the modes and for the scale factor and thus they are independent of $V$ at $\tau>\tau_0$.
  Furthermore, they are sufficient to obtain a bound for $Q_s[V]$ using the bounds obtained in Proposition \ref{prop:modes},  \[ |\chi_k|^2\leq \frac{1}{2k_0} \mathrm{e}^{\frac{2(\tau-\tau_0)\|V\|_\infty}{k_0}}, \qquad | \delta \chi_k[V,W] | \leq \frac{(\tau-\tau_0)}{\sqrt{2k_0^3}} \mathrm{e}^{\frac{(\tau-\tau_0)\|V\|_\infty}{k_0}} \| W\|_\infty, \] hence, on any interval of time $[\tau_0,\tau]$, from \eqref{eq:QSV} we have that \[ \| Q_s[V] \|_\infty \leq C \mathrm{e}^{\frac{2(\tau-\tau_0)\|V\|_\infty}{w_0}}, \qquad \|\delta Q_s[V,W]\|_\infty \leq C \mathrm{e}^{\frac{2(\tau-\tau_0)\|V\|_\infty}{w_0}}  \| W\|_\infty, \] where $w_0$ equals $\Omega_k(\tau_0)$ evaluated at $k=0$ and where a suitable constant $C$ is chosen.
 
  The estimates obtained above for $A$ and $B$ are not sufficient to get the desired bounds for $Q_s^d$. To get further control on these coefficients, we observe that \[ f_3\doteq \partial_\tau |\zeta_k|^2 - \partial_\tau |\chi_k|^2 +\frac{{V}'(\tau_0)}{4k_0^3} \cos(2k_0(\tau-\tau_0)) \] is $L^1(\mathbb{R}^3,\dif\overline{k})$ at $\tau_0$, because the initial conditions for the state are chosen in such a way that $\partial_\tau |\zeta_k|^2-\partial_\tau C^{\mathcal{H}}_{\phi^2}$ is absolutely $k$-integrable (see e.g. \eqref{eq:init-regular1}) and because, in the proof of Proposition \ref{prop:decomposition}, we proved that $\partial_\tau |\chi_k|^2 -\frac{{V}'(\tau_0)}{4k_0^3} \cos(2k_0(\tau-\tau_0)) - \partial_\tau C^\mathcal{H}_{\phi^2}$ is $L^1$ too. See also the Remark \ref{re:subtraction}. Notice that $f_3$ can be further expanded as
  \begin{equation}
    \label{eq:AB1}
    \begin{aligned}
      f_3 = & 2|B|^2\partial_\tau \left(\overline{\chi}_k(\chi_k-\chi_k^0)+ (\overline{\chi}_k - \overline{\chi}_k^0) \chi_k^0 \right) \\
            & +
      4\text{Re}(A\overline{B}(\chi_k-\chi_k^0){\chi_k}')+ 4\text{Re}(A\overline{B}\chi_k^0{(\chi_k-\chi_k^0)}')
      \\
            & +
      4\text{Re}(A\overline{B}\chi_k^0{\chi_k^0}')
      +\frac{{V}'(\tau_0)}{4k_0^3}\cos(2k_0 (\tau-\tau_0)).
    \end{aligned}
  \end{equation}
  From the previous discussion we know that 
  \begin{equation}
    \label{eq:AB1-0}
    \begin{aligned}
      f_3(\tau_0) = 2\text{Re}(A\overline{B}\mathrm{i}\mathrm{e}^{\mathrm{i}2k_0\tau_0})
      +\frac{{V}'(\tau_0)}{4k_0^3}
    \end{aligned}
  \end{equation}  
  is absolutely $\vec{k}-$integrable. Hence, the imaginary part of $A\overline{B}\mathrm{e}^{\mathrm{i}2k_0\tau_0}$ equals $V'(\tau_0)/8k_0^3$ up to an absolutely $\vec{k}-$integrale  function.  Arguing as before we also have that 
  $|\zeta_k|^2-C^\mathcal{H}_{\phi^2}$ is absolutely $\vec{k}$-integrable at $\tau_0$ and the same holds for $f_4\doteq (|\zeta_k|^2-C^\mathcal{H}_{\phi^2})/\Omega_k$. Hence at $\tau_0$ we get the absolutely $\vec{k}-$integrable function
  \[
  f_4(\tau_0) = 
  2 \left( |B|^2 \frac{|{\chi_k^0}'|^2}{k_0} + \text{Re} \left( A\overline{B}\frac{{{\chi_k^0}'}^2}{k_0} \right)\right)
  =
   \left( |B|^2  - \text{Re} \left( A\overline{B}\mathrm{e}^{\mathrm{i}2k_0\tau_0} \right)\right).
  \]
  Since $|AB|/k_0$ is absolutely $\vec{k}-$integrable and $|A|^2=1+|B|^2$ we also have that $|B|^2$ is absolutely $\vec{k}-$integrable, and thus the same holds for $\text{Re} \left( A\overline{B}\mathrm{e}^{\mathrm{i}2k_0\tau_0}\right)$.
  We conclude that $A\overline{B}\mathrm{e}^{\mathrm{i}2k_0\tau_0}$ equals $\mathrm{i} V'(\tau_0)/8k_0^3$ up to an absolutely integrable $\vec{k}-$function. This also means that $A$ and $B$ are such that
  \begin{equation}
    \label{eq:f4}
    f_5\doteq 4\text{Re}(A\overline{B}\chi^0_k{\chi^0_k}') +\frac{{V}'(\tau_0)}{4k_0^3}\cos(2k_0 (\tau-\tau_0))
  \end{equation}
  is absolutely $\vec{k}-$integrable at any time. To control the integrability of $f_3$ at later time we now evaluate the time evolutions of the modes $\chi_k$. From the \eqref{eq:chi-chi0} of Proposition \ref{prop:modes} we have that the first contribution in \eqref{eq:AB1} is in $L^1$ at any time because ${|B|^2}/{k_0}$ is $L^1$ and
   the second and the third contributions of $f_3$ in \eqref{eq:AB1} are also elements of $L^1$ because ${|{A}
    B|}/{k_0}$ is $L^1$, $|\chi'| \leq \sqrt{k_0} C \mathrm{e}^{\|V\|_1/w_0}$ and $|\chi| \leq C \mathrm{e}^{\|V\|_1/w_0}/\sqrt{k_0}$. Finally, also the contribution $f_5$ in \eqref{eq:f4} of $f_3$ in \eqref{eq:AB1} is integrable. We can now rewrite $Q_s^d[V]$ 
  \begin{align*}
    Q_s^d[V] & = \frac{1}{(2\pi)^3} \int_{\mathbb{R}^{3}}\left( 2 |B|^2 \partial_\tau |\chi_k|^2 + 4\text{Re} \left( A\overline{B} \chi_k\partial_\tau \chi_k\right) +\frac{{V}'(\tau_0)}{4k_0^3} \cos(2k_0(\tau-\tau_0)) \right) \dif \vec{k} \\  & = \frac{1}{(2\pi)^3} \int_{\mathbb{R}^{3}}\left( 2 |B|^2 \partial_\tau \left(\overline{\chi_k} (\chi_k-\chi_k^0) \right) + 2 |B|^2 \partial_\tau \left( (\overline{\chi_k} -\overline{\chi_k^0})\chi_k^0 \right) + 4\text{Re} \left( A\overline{B} \chi_k (\chi_k-\chi_k^0)' \right) \right.
    \\
             & \quad\left.+ 4\text{Re} \left( A\overline{B} (\chi_k-\chi_k^0) {\chi_k^0}' \right)
    + 4\text{Re} \left( A\overline{B} \chi_k^0{\chi_k^0}' \right)
    +\frac{{V}'(\tau_0)}{4k_0^3} \cos(2k_0(\tau-\tau_0)) \right) \dif \vec{k},
  \end{align*}
  where we have used the relation $|\chi^0_k|'=0$.
  Recalling the estimates \eqref{eq:chi-chi0}, having shown the integrability of $f_3$ at any time in \eqref{eq:AB1} and knowing how to control  $\chi_k$ and its functional derivatives, the statements of the Proposition follow.
\end{proof}

Notice that the coefficients $A$ and $B$ introduced in the proof of the previous proposition are constructed with $a(\tau_0),a'(\tau_0),a''(\tau_0)$ and $a^{(3)}(\tau_0)$ only. This implies that the constants used in the estimates could depend on the initial conditions and hence do not contain derivatives higher than the third order.

Finally, we observe that
\begin{equation}
  \label{eq:W0}
  Q_0 = \frac{m^2a^2 +(\xi-\frac{1}{6})Ra^2}{8\pi^2}\log(\frac{w(\tau_0)}{a}) - \frac{w(\tau_0)^2}{16\pi^2} +\alpha_1 a^2m^2 + \alpha_2 \left(\xi-\frac{1}{6}\right) a^2 R
\end{equation}
is a function of $a$ and $a''$ which is differentiable if $a>0$. Similarly, we observe that
\begin{equation}
  \label{eq:W0prime}
  \begin{aligned}{Q^d_0} = &
    \left( \frac{a^3H(m^2+(\xi-\frac{1}{6})R)}{4\pi^2} + \frac{a^2(\xi-\frac{1}{6})R'}{8\pi^2}\right)\log(\frac{w(\tau_0)}{a})
    - \frac{a^3 H (m^2 +(\xi-\frac{1}{6})R)}{8\pi^2}
    \\
              & +  \alpha_1 m^2 \partial_\tau (a^2)  +  \alpha_2 \partial_\tau (a^2R)
  \end{aligned}
\end{equation}
which is a function of $a$ and its derivative up to the third order.
This function is differentiable if $a>0$.

\section{Integration of the semiclassical Einstein equation}\label{se:SEE-integration}

In this section we consider the semiclassical Einstein equation \eqref{SCE} in the case of cosmological backgrounds.
Hence, since $\nabla^a \langle \wick{T_{ab}}\rangle_\omega =0$, we just need to consider the semiclassical version of the two equations \eqref{trace-equation} and \eqref{00-equation} where the expectation values $\langle \wick{T}\rangle_\omega = g^{ab}\langle \wick{T_{ab}}\rangle_\omega$ and $\langle\wick{\varrho}\rangle_\omega = \langle \wick{T_{00}}\rangle_\omega$ are used at the place of the corresponding classical quantities.
The expectation values are computed in a quasifree state $\omega$ which was introduced in \eqref{Gaussian-state} and it is characterized by the initial conditions \eqref{eq:init-cond-state} satisfying the regularity conditions \eqref{eq:init-regular1} and \eqref{eq:init-regular2}.

Thanks to the discussions of section \ref{se:states} and \ref{se:pointsplitting}, we have that the semiclassical Einstein equation on FLRW spacetimes is a dynamical problem for the scale factor $a(\tau)$ and for the state $\omega$ described by the following system of equations
\begin{equation}
  \label{SCE-trace}
  \begin{cases}
    - R(a,{a''}) + 4\Lambda = 8\pi G \expvalom{T} \at a,{a'},{a''},a^{(3)},a^{(4)} \ct, \\
    G_{00}(\tau_0) - a^2 \Lambda = 8\pi G \expvalom{T_{00}}\at a_0,{a_0'},{a_0''},a_0^{(3)}\ct,
  \end{cases}
\end{equation}
equipped with some initial conditions for $a$ and for $\omega$.
The \emph{initial conditions} for the scale factor $a$ are fixed at $\tau=\tau_0$ and consist of $(a(\tau_0),a'(\tau_0),a''(\tau_0),a^{(3)}(\tau_0)) = (a_0,a'_0,a''_0,a^{(3)}_0)$ while those for the state are given in terms of the functions $\Phi$, $E$ and $s$ introduced in \eqref{eq:init-cond-state}.

In writing this system of equations we have used the form of $\expvalom{T}$ presented in \eqref{traceT} which we recall here: \[ \expvalom{T} = \left(3\left(\xi - \frac{1}{6}\right) \square - m^{2}\right)\expvalom{\phi^2} + T_{A}  + \beta_1 m^{4} +\beta_2 m^{2} R +\beta_3 \square R.
\]
The anomalous term $T_A$ coincides to $[v_1]/(4\pi^2)$ given in \eqref{traceT}
up the renormalization freedom and has the form
\[
  T_{A} = \frac{1}{4 \pi^{2}}
  \left(\frac{(6\xi - 1)^2 R^2}{288}  + \frac{R_{abcd}
    R^{abcd} - R_{ab}R^{ab}}{720}	\right), \]
while $\beta_i$ are renormalization constants which are universal and thus they are fixed once and forever \cite{Brunetti:2001dx}.
Regarding their physical meaning, we notice that changing $\beta_1$ corresponds to a renormalization of the cosmological constant and changing $\beta_2$ corresponds to a renormalization of Newton's gravitational constant $G$, while the remaining constant $\beta_3$ has no classical interpretation.

We observe that $\expvalom{T}$ contains an explicit contribution which depends locally on the fourth order derivative of $a$. Furthermore, as we shall carefully see later, there is a fourth order derivative contribution which is non-local through $\square \expvalom{\phi^2}$. This non-local contribution is due to $\mathcal{T}_{\tau_0}[V]$ introduced in \eqref{T-Vk} within Proposition \ref{prop:decomposition}.

The second equation needs to be fulfilled at $\tau=\tau_0$ and it is thus a \emph{constraint} on the initial values for the state. We have actually seen in section \ref{se:states} that $\expvalom{T_{00}}$ can be constructed only with the derivatives of the scale factor up to the third order and with the $\Phi, E$ and $s$, which specify the initial values for the state once the renormalization constants are fixed. We may thus introduce the following Definition.
\begin{definition}\label{def:compatible}
	Consider a flat cosmological spacetime, whose scale factor $a$ is characterized by the initial conditions $(a_0,a'_0,a''_0,a^{(3)}_0)$ fixed at $\tau_0$ with $a_0>0$. Let $\omega$ be a quantum state given in \eqref{Gaussian-state}. We say that $\omega$ is \emph{compatible with the initial conditions} if the initial constraint (the second equation in \eqref{SCE-trace})  \[ H(\tau_0)^2 =  \frac{8\pi G}{3}\langle \wick{\varrho} \rangle_\omega (\tau_0) + \frac{\Lambda}{3} \] is satisfied at $\tau=\tau_0$.
\end{definition}
We thus start proving that it is always possible to choose a state compatible with initial conditions.

\begin{proposition}
  \label{prop:Friedman-constraint}
  Let $a_0,a_0',a_0''$ and $a_0^{(3)}$ be initial conditions for $a$ and its derivatives at time $\tau=\tau_0$ chosen in such a way that $\Omega_k(\tau_0)^2$ given in \eqref{mode-eq} is strictly positive.
  Let the coupling to the curvature $\xi\neq \frac{1}{6}$.
  Fix the renormalization constants.
  It is possible to select initial conditions $\Phi, E$ and $s$ in \eqref{eq:init-cond-state} that fix the quantum state $\omega$ given in \eqref{Gaussian-state} in such a way that the state is sufficiently regular in the sense of Definition \ref{def:sufficiently-regular} and
  compatible with initial conditions in the sense of Definition \ref{def:compatible}.
\end{proposition}
\begin{proof}
  As shown in \eqref{eq:rho} and in view of the form of $C^\mathcal{H}_\varrho$, in order to evaluate $\varrho$ in the state $\omega$, only third derivatives of $a$ are necessary. In particular, let us fix some initial condition for the state \eqref{eq:init-cond-state} which are described by $\Phi, E $ and $s$ sufficiently regular (namely they satisfy \eqref{eq:init-regular1} and \eqref{eq:init-regular2}) and writing $\zeta_k = \rho \mathrm{e}^{\mathrm{i}\theta}$, we get the following real finite result
  \begin{align*}
    \langle\wick{{\varrho}}\rangle_\omega(\tau_0) = & \frac{1}{(2\pi)^3 a(\tau_0)^4} \int_{\mathbb{R}^{3}} \aq \frac{1}{2}\at (1-6\xi) (\rho'- aH \rho)^2 + 6\xi {\rho'}^2+ (k^2+a^2m^2) \rho^2 \ct - C^\cH_{\varrho}(\tau_0, k) \cq \dif \vec{k} + \\  & -\frac{H(\tau_0)^4}{960\pi^2} +\left(\xi-\frac{1}{6} \right)^2\frac{3 H^2(\tau_0)R(\tau_0)}{8\pi^2}  +\tilde{\beta}_1 m^4 - \tilde{\beta}_2 m^2 {G_{0}}^0(\tau_0)+(\tilde{\beta}_3-\frac{\tilde{\beta}_4}{3}) {I_{0}}^0(\tau_0),
  \end{align*}
  where $\rho = \sqrt{\Phi}$ and $\rho' = s\sqrt{E-\frac{1}{4\Phi}}$.

  We show that changing initial conditions for the state, $\langle\wick\varrho\rangle_\omega$ can assume all values over the real line.
  We discuss in some detail the case $H\neq 0$, the other cases can be discussed similarly.
  The initial conditions $\tilde{E},\tilde{\Phi}$ and $\tilde{s}$ for a new state $\tilde\omega$ are chosen in such a way that 
  \[ 
  	\tilde{s} = s , \qquad \tilde{\Phi} = \Phi, \qquad    \tilde{E}= \left(\sqrt{E-\frac{1}{4 \Phi}}  + \frac{C}{k^{\frac{5}{2}}}\Pi_{[p_1,p_2]}(|k|) \right)^{2} + \frac{1}{4 \Phi},
  \]    
  where $C$ is a constant, $\Pi_{[p_1,p_2]}$ is the characteristic function of the interval $[p_1,p_2]$ of the positive real line. Notice that for every choice of the parameter $p_1,p_2$ the constraint
  \[ 
  	\tilde{E} -\frac{1}{4 \tilde{\Phi}} \geq 0 
  \] 
  is fulfilled. Hence, in this way
  \[
  	\tilde{\rho} = \rho, \qquad \tilde{\rho}'= {\rho}'  + s\frac{C}{k^{\frac{5}{2}}}\Pi_{[p_1,p_2]}(|k|).
  \]
  We observe that the correction is such that
  \begin{align*}
    \langle\wick{\varrho}\rangle_{\tilde\omega} & = \langle\wick{\varrho}\rangle_\omega + \frac{1}{(2\pi)^3}\frac{1}{2a^4}
    \int \left( \frac{s^2C^2}{k^5} + s\frac{2C}{k^{\frac{5}{2}}} (\rho'- (1-6\xi )aH \rho ) \right)  \Pi_{[p_1,p_2]}   \dif\vec{k}  \\
                                                & = \langle\wick{\varrho}\rangle_\omega + \frac{1}{2\pi^2}\frac{s^2C^2}{4a^4}  \left( \frac{1}{p_1^2} - \frac{1}{p_2^2} \right) + \frac{1}{(2\pi)^3}\frac{s2C}{2a^4}\int \frac{1}{k^{\frac{5}{2}}} (\rho'- (1-6\xi )aH \rho )  \Pi_{[p_1,p_2]}   \dif\vec{k} \\
                                                & = \langle\wick{\varrho}\rangle_\omega + \frac{1}{2\pi^2}\frac{s^2C^2}{4a^4}  \left(\frac{1}{p_1^2} - \frac{1}{p_2^2} \right) + \frac{s 2C}{2\pi^2}\frac{(6\xi-1 )H}{2a^3}  \log\left( \frac{p_2}{p_1}\right) + \frac{s 2C}{2\pi^2}\frac{1}{2a^4}\int_{p_1}^{p_2} O\left( \frac{1}{k^2} \right)  \Pi_{[p_1,p_2]}   \dif k.
  \end{align*}
  In the last equality we used the following decay properties of the state
  \[ 
  	\rho = \frac{1}{\sqrt{2k}} \left(1+O \left(\frac{1}{k}\right) \right), \qquad \rho' = \frac{1}{\sqrt{2k}} \left(O \left(\frac{1}{k}\right) \right),
  \]
  which hold thanks to the regularity conditions about the $L^1$-integrability that we have imposed on modes $\zeta_k, \zeta'_k$ at the initial time $\tau_0$. Actually, the large-$k$ behaviour of $\rho$ can be obtained from the first condition in \eqref{eq:init-regular1}, whereas the one for $\rho'$ can be derived either from the second condition in \eqref{eq:init-regular1} or from \eqref{eq:init-regular2}. In particular, the contribution proportional to $\log\left( \frac{p_2}{p_1}\right)$ follows from the leading contribution $1/\sqrt{2k}$ inside $\rho$.
  Hence, modifying $p_2$ and the sign of the constant $C$, we have that $\langle\wick{\varrho}\rangle_{\tilde\omega} - \langle\wick{\varrho}\rangle_\omega$  and thus also $\langle\wick{\varrho}\rangle_{\tilde\omega}$ can take all the values of the real line because $\log(p_2/p_1)$ diverges for $p_2\to\infty$ while the other contributions stay finite. We finally observe that the asymptotic behaviour of $\tilde{E}$ and $\tilde{\Phi}$ for large $|k|$ equals that of ${E}$ and ${\Phi}$, hence, the state $\tilde{\omega}$ is sufficiently regular in the sense of Definition \ref{def:sufficiently-regular} because this property holds for $\omega$.
\end{proof}

From now on the state $\omega$ will be considered to be fixed and chosen to be sufficiently regular according to Definition \ref{def:sufficiently-regular} and compatible with the initial conditions for the scale factor $a$ in the sense of Definition \ref{def:compatible}, hence
the second equation in \eqref{SCE-trace} is fulfilled. We are thus left with the analysis of the first equation in \eqref{SCE-trace}. We discuss it in the case of non-conformal coupling $\xi \neq 1/6$, i.e., when the higher order derivative terms $\square R$ and $\square \expvalom{\phi^2}$ cannot be avoided (the conformal coupling case $\xi = 1/6$ with suitable initial conditions given at finite conformal time and with a suitable choice of renormalization constants has already been solved in \cite{Pinamonti:2010is,Pinamonti:2013wya}). In order to do it we shall rewrite the trace of the semiclassical Einstein equations as a system of two equations. The first one is a non-homogeneous free Klein-Gordon equation (with imaginary mass if $\xi<\frac{1}{6}$) on FLRW spacetime while the second is an equation which involves the expectation value of $\wick{\phi^2}$. To this avail, in the first step we show the following: 

\begin{proposition}
  \label{pr:SCE-wave}
  Consider a generic spacetime $(M,g)$, let $R_{abcd}$, $R_{ab}$ and $R$ be respectively the Riemann tensor, the Ricci tensor and the Ricci scalar of the spacetime.
  Over this spacetime consider a real scalar field with mass $m$ and with coupling to the scalar curvature $\xi\neq 1/6$ .
  Let $\omega$ be a state for this real scalar field.
  The traced semiclassical Einstein equation (the first equation in \eqref{SCE-trace}) can be written as the following system of equations
  \begin{equation}
    \label{SCE-wave}
    \begin{cases}
      (-\square + M_c ) F = S, \\
      \expvalom{\phi^2} - c_\xi R = F,
    \end{cases}
  \end{equation}
  where
  \begin{equation}
    \label{eq:C-M-constants}
    c_\xi \doteq \frac{\beta_3}{3(1/6-\xi)}, \qquad
    M_{c} \doteq -\frac{m^2}{3(1/6-\xi)}
  \end{equation}
  and $S$ is a function of the derivatives of $a$ up to the second order
  \begin{equation*}
    S \doteq \frac{1}{3(\xi-1/6)} \mleft(  \beta_1 m^4 -\frac{\Lambda}{2\pi G} + \frac{R}{8\pi G} + \beta_2m^2 R + \beta_3 M_c R + \frac{(6\xi -1)^2 R^2}{1152\pi^2} + \frac{R_{abcd}R^{abcd} - R_{ab}R^{ab}}{2880\pi^2} \mright). 
  \end{equation*}
\end{proposition}
\begin{proof}
  The proof consists only of some algebraic manipulations of the traced semiclassical equation with $\expvalom{T}$ given by \eqref{traceT}; in particular, \eqref{SCE-trace} reads 
  \[
    \begin{aligned}
      -R + 4\Lambda & = 8\pi G \biggl( -3(1/6 - \xi) \square \expvalom{\phi^2} - m^2\expvalom{\phi^2} + \beta_3 \square R + \beta_1 m^4 + \beta_2 m^2 R \\
      				& \quad + \frac{(6\xi -1)^2R^2}{1152\pi^2} + \frac{R_{abcd}R^{abcd} - R_{ab}R^{ab}}{2880\pi^2} \biggr).
    \end{aligned}
  \]
  Eq. \eqref{SCE-wave} follows by regrouping the state-dependent terms and isolating all the terms depending on $\square$.
  Notice that all the functions $c_\xi,S(\xi,m^2,R,R_{ab},R_{abcd})$ and $M_{c}$, whose sign depends on the value of the parameter $\xi$, are well-defined for $\xi \neq 1/6$.
\end{proof}

We now specialize this discussion for a FLRW spacetime. In that case the geometric quantities depend only on $\tau$. Furthermore, if the state for the quantum matter is compatible with the cosmological principle, namely it is homogeneous and isotropic, also the expectation value of $\wick{\phi^2}$ depends only on the conformal time. Hence, in that case $F$ and $S$ are functions of $\tau$, and \eqref{SCE-wave} is such that
\begin{equation}
  \label{eq:ODE}
  P_c F = S,
\end{equation}
where
\begin{equation}
  \label{Pc}
  P_c \doteq \frac{1}{a^3(\tau)} \at \pa_\tau^2 + a^2(\tau) M_c - \frac{1}{6}a^2(\tau)R \ct a(\tau).
\end{equation}
Notice that \eqref{eq:ODE} is a second order differential equation. We can solve it observing that $P_c$ admits unique advanced and retarded fundamental solutions denoted by $\Delta_A^c$ and $\Delta_R^c$ respectively. Hence, if we equip equation \eqref{eq:ODE} with suitable initial conditions at $\tau=\tau_0$, we will have a unique solution that can be written in terms of the retarded fundamental solution.
We collect this observation in the following proposition.        

\begin{proposition}
  \label{prop:Pc}
  Let $a\in C^2[\tau_0,\tau] $ be a real positive function, let $h\in C[\tau_0,\tau]$ and $(f_0,{f}'_0)\in \mathbb{R}^2$.
  The following problem in $C^2[\tau_0,\tau]$
  \begin{equation}
    \label{KG-equation-c}
    \begin{cases}
      P_c f = h, \\ (f,{f}')(\tau_0) = (f_0, {f}'_0),
    \end{cases}
  \end{equation}
  where $P_c$ given in \eqref{Pc} admits a unique solution $f$ on $C^2([\tau_0,\tau];\mathbb{R})$ given in terms of $h$ and the initial data. Furthermore,  $\tilde{f} = a f$ depends linearly on the initial data $(\tilde{f}_0,\tilde{f}'_0)$, and on $\tilde{h} = a^3 h$.
  The following estimate holds
  \begin{equation}
    \label{C-ineq}
    \| \tilde{f} \|_\infty \leq \left( \| \beta \|_\infty + (\tau-\tau_0)^2 \|\tilde{h}\|_\infty \right)\exp \left( (\tau-\tau_0)^2\| \tilde{W} \|_\infty\right),
  \end{equation}
  where
  \[
    \tilde{W} = a^2 M_c - \frac{1}{6}a^2R, \qquad
    \beta(\tau) =  |\tilde{f}(\tau_0)| + (\tau-\tau_0)|{\tilde{f}'}(\tau_0)|.
  \]
  Furthermore $\tilde{f}$ depends continuously on $a$ and in particular denoting by $\delta \tilde{f}$ the functional derivative of $f$ with respect to infinitesimal changes $\delta a$ of $a$  
\[
 \|\delta \tilde{f}\|_\infty \leq (\tau-\tau_0)^2 \at \| \delta \tilde{W} \|_{\infty} \| \tilde{f}\|_{\infty} + \| \delta \tilde{h} \|_{\infty} \ct \exp((\tau - \tau_0)^2 \|\tilde{W}\|_\infty).
\] 
  All these norms are finite in any $[\tau_0,\tau]$ for all finite $\tau > \tau_0$.
\end{proposition}
\begin{proof}
  The problem \eqref{KG-equation-c} can equivalently be written as \[ \tilde{f}''  + \tilde{W} \tilde{f} = \tilde{h},   \] which is a second order linear non-homogeneous differential equations with regular coefficients. Since $a>0$, the problem \eqref{KG-equation-c} admits thus an unique solution. Once written in this way, the problem is of the form given in \eqref{eq:second-order-ode} we my thus apply the results of Lemma \ref{le:second-order-ODE} with $k=0$ from which we get \eqref{C-ineq}.
From Lemma \ref{le:second-order-ODE} and in particular from \eqref{eq:rewert-equation} we also have that 
\[ 
	\tilde{f}  =  - \int_{\tau_0}^\tau  (\tau-\eta)\tilde{W}(\eta)  \tilde{f}(\eta)  \dif \eta  + 
	\int_{\tau_0}^\tau  (\tau-\eta)  \tilde{h}(\eta)  \dif \eta  + \tilde{f}_0+ (\tau-\tau_0) \tilde{f}'_0. 
\]

Hence,
\begin{align*}
    |\delta \tilde{f}(\tau)| & \leq (\tau-\tau_0)  \int_{\tau_0}^\tau \at |\delta \tilde{W}(s)||\tilde{f}(s)| + |\tilde{W}(s)| |\delta \tilde{f}(s)| \ct \dif s + (\tau-\tau_0) \int_{\tau_0}^\tau |\delta \tilde{h}(s)| \dif s \\
                             & \leq \gamma(\tau)  + (\tau-\tau_0)  \int_{\tau_0}^\tau |\tilde{W}(s)||\delta \tilde{f}(s)| \dif s,
  \end{align*}
  where $\gamma(t) = (\tau - \tau_0)^2\at \| \delta \tilde{W} \|_{\infty} \| \tilde{f}\|_{\infty} + \| \delta \tilde{h} \|_{\infty} \ct$, and by Gr\"onwall lemma
  \begin{align*}
    |\delta \tilde{f}(\tau)| & \leq \gamma(\tau) \mathrm{e}^{(\tau-\tau_0) \int_{\tau_0}^\tau |\tilde{W}(s)| \dif s}.
  \end{align*}
  Thus the remaining statements follow.
\end{proof}

Hence we have the following
\begin{theorem}
  \label{theo:sol-SCE-wave}
  Consider a flat cosmological spacetime. Fix five constants $\tau_0$, $a_0$, ${a'}_0, {a''}_0, a^{(3)}_0$, with $a_0 >0$, corresponding at the initial data in $\tau_0$ for \eqref{SCE-wave}, namely $a_0 = a(\tau_0), {a'}_0 = {a'}(\tau_0), {a''}_0 = {a''}(\tau_0), a^{(3)}_0 = a^{(3)}(\tau_0)$. Fix the renormalization constants and let $\xi \neq \frac{1}{6}$.
  Select a state $\omega$ as in \eqref{Gaussian-state} characterized by \eqref{eq:init-cond-state} which is compatible with these initial conditions and thus the second equation in \eqref{SCE-wave} holds.
  It exists a unique $F$ given in terms of $S$ and the initial data which is a solution of the first equation \eqref{SCE-wave} and which depends continuously on the initial data $a_0, a_0', a_0'', a_0^{(3)}$, on the initial data for the state $\omega$ and on $a$.
  Moreover,
  \begin{equation}
    \label{C-ineq-F}
    \| \tilde{F} \|_\infty \leq
    \left( \| \mI_0 \|_\infty + (\tau-\tau_0)^2 \|\tilde{S}\|_\infty \right)\exp \left( (\tau-\tau_0)^2\| \tilde{W} \|_\infty\right),
  \end{equation}
  \begin{equation}
    \label{C-ineq-deltaF}
    \|\delta \tilde{F}\|_\infty \leq (\tau-\tau_0)^2 \at \| \delta \tilde{W} \|_{\infty} \| \tilde{F}\|_{\infty} + \| \delta \tilde{S} \|_{\infty} \ct \exp((\tau - \tau_0)^2 \|\tilde{W}\|_\infty),
  \end{equation}
  where $\tilde{F} = {F}a$, $\tilde{S} = S a^3$ and $\| \mI_0 \|_\infty \doteq |\tilde{F}(\tau_0)| + (\tau-\tau_0) | {\tilde{F}'}(\tau_0)|$ corresponds to the finite norm of $\tilde{F}(\tau)$ at the initial time $\tau_0$ depending on the initial data of $a$ and the state.

  The system of equations \eqref{SCE-trace} reduces to  the second equation in \eqref{SCE-wave}
  \begin{equation}
    \label{SCE-state}
    \expvalom{\phi^2} - c_\xi R =
    {F}(a,R),
  \end{equation}
  where $F$ is the unique solution of \eqref{SCE-wave} given in terms of $S$ and the initial data of $a$ and of the state $\omega$.
\end{theorem}

\begin{proof}
  We start observing that, thanks to Proposition \ref{prop:Friedman-constraint}, the state can be chosen to be compatible with the first Friedman equation at $\tau_0$.
  The initial data on the scale factor and its first three derivatives allow to construct the corresponding initial data for $F$:
  \begin{equation}
    \label{eq:init-cond-F}
    \begin{aligned}
      \tilde{F}(\tau_0)    & = a_0 \left( \expval{\wick{\phi^2}}_\omega(\tau_0) - c_\xi R_0\right),                                                                                         \\
      {\tilde{F}}'(\tau_0) & = {a}'_0 \left( \expval{\wick{\phi^2}}_{\omega}(\tau_0) - c_\xi R_0\right) + a_0 \left( \pa_\tau \expval{\wick{\phi^2}}_\omega(\tau_0) - c_\xi{R}'_0 \right),
    \end{aligned}
  \end{equation}
  where the expectation values $\expval{\wick{\phi^2}}_\omega$ and $\pa_\tau \expval{\wick{\phi^2}}_{\omega}$ are evaluated at $\tau=\tau_0$, respectively.
  These initial data depends on the modes initial data \eqref{eq:init-cond-state} at $\tau=\tau_0$ and on the initial data of the geometry $a_0, a_0', a_0''$ and $a_0^{(3)}$.
  The unique solution $F$ of the first equation \eqref{SCE-wave} which satisfies the initial conditions \eqref{eq:init-cond-F} and its bounds are obtained in Proposition \ref{prop:Pc}.
  We have thus partially integrated the system of equation \eqref{SCE-wave} and we are left with the second equation in \eqref{SCE-wave}.
\end{proof}

\section{Semiclassical Einstein equation as a fixed point equation}\label{se:SEE-fixed-point}

In section \ref{se:SEE-integration} and in particular thanks to Theorem \ref{theo:sol-SCE-wave} we have reduced the problem of finding solutions of the semiclassical Einstein equation on flat cosmological backgrounds to the problem of finding solutions of \eqref{SCE-state}, which satisfy the desired initial conditions.
We now prove the existence of a unique solution of \eqref{SCE-state}, and thus of \eqref{SCE-trace}, on a small interval of conformal time just after the initial time $\tau = \tau_0$.
In order to have control on the third order derivative of $a$ and to be able to impose $a^{(3)}(\tau_0)=a_0^{(3)}$ we study the time derivative of the equation
\begin{equation}
  \label{eq:time-derivative-phi2}
  \partial_\tau \left( a^2( \expvalom{\phi^2} - c_\xi R - F) \right)  = 0.
\end{equation}
Notice that this equation is equivalent to \eqref{SCE-state} because at $\tau_0$ the equation without derivatives
\[ 
	\left. \expvalom{\phi^2} - c_\xi R - F \right|_{\tau_0} =0
\]
holds thanks to the choice of initial condition made for $F$ in \eqref{eq:init-cond-F}.

Adopting the same strategy presented in \cite{Pinamonti:2010is} and \cite{Pinamonti:2013wya}, we show that, having fixed initial conditions for $a$ and having chosen a state $\omega$ compatible with these initial conditions thanks to the result of Proposition \ref{prop:Friedman-constraint} equation \eqref{SCE-state} can be viewed as a fixed point equation
\begin{equation}
  \label{fixed-point-equation}
  X'=\mathcal{C}[X'], \qquad X'\in C[\tau_0,\tau_1],
\end{equation}
where $X$ can be obtained from $X'$ by direct integration with the condition $X(\tau_0)=X_0$, furthermore $X$ is directly related to the scale factor of the spacetime
\begin{equation}
  \label{X}
  X = \frac{1}{6}a^2R = a''/a.
\end{equation}
The initial conditions for $X_0$ and $X'_0$ are then fixed by the initial conditions of the scale factor $a$
\[
	X_0 \doteq X(\tau_0)=\frac{a''_0}{a_0}, \qquad X'_0 \doteq X'(\tau_0) = \frac{a_0'''}{a_0} -\frac{a_0''a'_0}{a_0^2}.
\]
Since, $\mathcal{C}$ is a suitable map acting on $C[\tau_0,\tau_1]$ which is a Banach space when equipped with the uniform norm, to prove existence and uniqueness of solutions of the analyzed system, we shall show that the map $\mathcal{C}$ is a contraction when restricted on a suitable compact subset
\begin{equation}
  \label{eq:ball-delta}
  \mathcal{B}_\delta \doteq \left\{X'\in C[\tau_0,\tau_1]\mid \|X'-X'_0\| \leq \delta\right\}
\end{equation}
when $\tau_1-\tau_0$ is sufficiently small. Thus, the existence and uniqueness of the solution descends from the application of the Banach fixed point theorem.

The function $X$ is exactly the curvature-like quantity entering the second Friedmann equation or \eqref{trace-equation} and is contained both in the scale factor constructed as the unique solution of
\[
  \begin{cases}
    a''=Xa,          \\
    a'(\tau_0)=a'_0, \\
    a(\tau_0)=a_0
  \end{cases}
\]
and in the potential $V$ via the definition \eqref{potential}
\begin{equation}
  \label{potential-2}
  V(\tau) = m^2(a^2 - a_0^2) + (6\xi - 1) \at X - X_0 \ct.
\end{equation}
Hence, from now on we shall view both of them as functionals of $X$ and thus of $X'$ because $X(\tau_0)=X_0$.
Some useful inequalities satisfied by $a[X]$ are given in Lemma \ref{le:a[x]}.

Different from the case of conformal coupling \cite{Pinamonti:2010is, Pinamonti:2013wya}, we need to better analyze $\expval{\wick{\phi^2}}_\omega$.
Actually, according to Proposition \ref{prop:decomposition}, in the state dependent part of this expectation value there is a non-local term (the linear operator $\operatorT[V]$ given in equation \eqref{T-Vk}) which depends on derivatives of the scale factor higher than the second and in $\partial_\tau\expval{\wick{\phi^2}}$ there is a similar term (the linear operator $\operatorT[V']$ applied on $V'$) which depends on fourth order derivatives of the scale factor.
The presence of these non-local higher order derivatives forbids to solve \eqref{eq:time-derivative-phi2} directly and actually this is the source of the regularity issues found in the formulation of the semiclassical Einstein equation on cosmological spacetimes as a dynamical system \cite{Gottschalk:2018kqt}.
We shall show how to deal with these higher order derivatives which cannot be controlled on $\mathcal{B}_\delta$.
However, the equation $h=\operatorT[V']$ admits an inversion formula continuous in the norm of $\mathcal{B}_\delta$, namely with respect to $\| X' \|_\infty$.
Thus, we can recover control of its continuity in $\mathcal{B}_\delta$.

Thus, in order to prove the existence of solutions $a$ of \eqref{SCE-state} which satisfy the desired initial conditions, a careful analysis of each term of $\partial_\tau (a^2 \langle\wick{\phi^2}\rangle_{\omega})= Q^d_0 + \operatorT_{\tau_0}[V'] + Q_f^d + Q^d_s$ given in \eqref{eq:Dphi2-decomposition} is necessary.
We thus have the following proposition which descends directly from the the decomposition of the state \eqref{eq:Dphi2-decomposition}, furthermore, the continuity of some of the involved operators as functional operators of $V$ and hence $X$ in $\cB_\delta$ is guaranteed by the results of Theorem \ref{theo:sol-SCE-wave} and of Proposition \ref{prop:decomposition}.

\begin{proposition}
  \label{prop:SCE-V}
  Fix the initial data $a_0,a'_0, a''_0$ and $a_0^{(3)}$ for $a$, assume that $\Omega_k^2$ given in \eqref{mode-eq} is strictly positive,  the semiclassical equation \eqref{eq:time-derivative-phi2} can be expressed in terms of the potential $V$ \eqref{potential-2}, $R$ and $a$ as
  \begin{equation}
    \label{SCE-V}
    Q^d_0 + \operatorT_{\tau_0}[V'] + Q_f^d + Q^d_s  =   \partial_\tau \left(a^2 c_\xi R + a^2 {F}(a,R) \right),
  \end{equation}
  where $Q^d_0$ is given in \eqref{eq:W0}, $Q^d_f$ and $\operatorT_{\tau_0}[V']$ are introduced in Proposition \ref{prop:decomposition} and $Q^d_s$ is given in Proposition \ref{prop:FV}.
  $F$ is given in \eqref{SCE-state} and its properties are stated in Theorem \ref{theo:sol-SCE-wave}.
\end{proposition}

Inside Eq.~\eqref{SCE-V}, the most problematic contribution is the operator map $\operatorT_{\tau_0}[V']$ defined in \eqref{T-Vk}, which contains the essence of the non-local nature of the semiclassical equation \eqref{SCE}, because to compute $\operatorT_{\tau_0}[V'](\tau)$ we need to know $V''(\eta)$ for every $\eta\in [\tau_0,\tau]$. Thus, this operators introduces the main difficulty of this equation: in fact, we shall prove that it is not continuous in $C(\aq\tau_{0}, \tau\cq)$ with respect to the uniform norm for any $\tau > \tau_0$.
In the next section we shall thus study the operator \[ \operatorT_{\tau_0}[f](\tau)= -\frac{1}{8\pi^2}\int_{\tau_0}^{\tau} {f'}(\eta) \log(\tau - \eta) \dif \eta  \] for $f\in C^{1}[\tau_0,\tau]$.

\begin{remark}\label{re:problem}
  It is of interest to point out the difference between equation \eqref{SCE-state} and the analogous equation of the conformal coupling case $\xi = 1/6$ \cite{Pinamonti:2013wya}:
  \begin{equation}
    \label{eq:conf-coupled}
    \frac{\dif H}{\dif\tau} = \frac{a}{H_{c}^{2}-H^{2}}\left(H^{4}-2 H_{c}^{2}
    H^{2} + 240 \pi^{2}\left(m^{2} \expvalom{\phi^2} + \beta m^{2} R + 4 \tilde{\Lambda}\right)\right),
  \end{equation}
  where $H_{c}^{2} \doteq 1440 \pi^{2} /(8 \pi \mathrm{G})=180 \pi / \mathrm{G}$ and $\tilde{\Lambda} \doteq \Lambda/(8\pi G)$.
  Actually, in the case of conformal coupling, it is possible to cancel the terms with derivatives higher than the second by a choice of the renormalization constants.
  Furthermore, in the conformally coupled case it is possible to obtain estimates for the renormalized Wick square which involves only up to first order derivatives of $a$ (see \cite{Pinamonti:2010is, Pinamonti:2013wya}).
  Equation \eqref{eq:conf-coupled} is written in normal form, namely, the term with the highest derivative of the scale factor is isolated at the left hand side of the equation.
  Thanks to this fact equation \eqref{eq:conf-coupled} can be solved directly.
  Instead, in the case of generic coupling, the presence of a non-local term which contains third order derivatives of the scale factor $a^{(3)}$ inside $\expvalom{\phi^2}$, through the linear operator $\operatorT_{\tau_0}[V]$ in \eqref{T-Vk},
     forbids to prove existence of solutions directly for \eqref{SCE-state}.
  In a similar way, the non-local term which contains fourth order derivatives of the scale factor $a^{(4)}$ inside $\partial_\tau(a^2\expvalom{\phi^2})$ through $\operatorT_{\tau_0}[V']$ forbids a direct analysis of 
\eqref{eq:time-derivative-phi2}.
\end{remark}

\subsection{Properties of the operator $\operatorT_{\tau_0}$ and its inversion formula}

In this section we shall analyze the properties of the operator $\operatorT_{\tau_0}$ defined in \eqref{T-Vk}. 
In particular, a way to overcome the problem of the loss of derivatives of $\operatorT_{\tau_0}[f]$ envisaged in Remark \ref{re:problem} is to study an inversion formula for $h=\operatorT_{\tau_0}[f]$ and to prove the continuity of the associated inverse operator $\operatorT_{\tau_0}^{-1}$. Preliminarily, we notice that, up to a translation of the function $f_{\tau_0}(x) = f(x+\tau_0)$, 
\[ 
\operatorT_{\tau_0}[f](x+\tau_0)= \operatorT_0[f_{\tau_0}](x).
\]
We shall thus study the following  operator $\operatorT : C_1^{\infty}[0,r] \rightarrow C[0,r]$, with $r>0$,
  \begin{equation}
    \label{T}
    \operatorT[f](x) \doteq -\frac{1}{8\pi^2}\int_\mathbb{R^+} f^\prime(y) \theta(x-y) \log(x-y) \dif y = -\frac{1}{8\pi^2}\int_{0}^{x} f^\prime(y) \log (x-y) \dif y
  \end{equation}
  which equals  $\operatorT_0$ on $C^{1}[0,\tau-\tau_0]$. Clearly $\operatorT$ is bounded in the $C^1$ sense.
Indeed, since $\log x$ is integrable in $x=0$, we find
\[
  {\|\operatorT[f]\|}_\infty
  \leq \sup_{x \in [0,r]} \int_0^x |f'(y)| \, |\log(x-y)| \dif y
  \leq {\|f'\|}_\infty \int_0^r |\log(x-y)| \dif y
  \lesssim {\|f'\|}_\infty,
\]
where we denote by ${\|X\|}_\infty$ the Banach norm on the space of continuous function $C[0,r]$.
However, $\operatorT$ is not bounded in the $C^0$ sense.
In fact, even if we assume smoothness and compact support, $\operatorT[f]$ cannot be bounded by ${\|f\|}_\infty$.

\begin{proposition}
  \label{prop:T}
  The restriction of $\operatorT$ to $C^\infty_0[0,r]$ is not bounded in the sense of $C[0,r]$.
\end{proposition}
\begin{proof}
  We consider the action of $\operatorT$ on a sequence of smooth functions $f_\varepsilon$ in the limit $\varepsilon \rightarrow 0^+$.
  For an arbitrary interval $[x_1,x_2] \subset (0,r)$, let
  \[
    f_\varepsilon(x) \doteq \theta_\varepsilon(x_1-x)-\theta_\varepsilon(x_2-x),
    \quad
    \theta_\varepsilon(x) \doteq \int_{-\infty}^{x} \frac{1}{\varepsilon} \varphi\at\frac{y}{\varepsilon}\ct \dif y,
  \]
  where $\varphi \in C^\infty_0[-1,1]$ is a positive mollifier (in particular, ${\|\varphi\|}_1=1$).
  By construction, for sufficiently small $\varepsilon$, we have $f_\varepsilon \in C^\infty_0[0,r]$ and ${\|f_\varepsilon\|}_\infty = 1$.
  However, for almost every $x$, we get
  \[
    \lim_{\varepsilon \to 0^+} \operatorT[f_\varepsilon](x) =
    \begin{dcases}
      -\log(x-x_1) + \log(x-x_2), & x > x_2,       \\
      -\log(x-x_1),               & x_1 < x < x_2, \\
      0,                          & x < x_1.
    \end{dcases}
  \]
  Hence, for all $M > 0$ there exists $\varepsilon^* > 0$ such that ${\|\operatorT[f_\varepsilon]\|}_\infty \geq M$ for all $\varepsilon \in (0,\varepsilon^*)$.
  Since ${\|f_\varepsilon\|}_\infty = 1$, this proves that $\operatorT$ is not continuous with respect to the uniform norm.
\end{proof}
For this reason, $\operatorT_{\tau_0}[V]$ in $\expvalom{\phi^2}$  and $\operatorT_{\tau_0}[V']$ in $\partial_\tau(a^2\expvalom{\phi^2})$ lose derivative and a direct analysis of \eqref{SCE-state} and of \eqref{eq:time-derivative-phi2}
is not possible.
A way to overcome this problem is to study an inversion formula for $h=\operatorT[f]$ for functions defined on the interval $[0,r]$.
Actually, we shall see that the inverse operator appearing in the inversion formula is more regular than $\operatorT$.
Hence, we shall use the analogous inversion formula for $\operatorT_{\tau_0}$ which is then obtained by a translation.

\begin{proposition}
  \label{prop:T-inverse}
   Consider $\operatorT$ introduced in \eqref{T}.
    The inversion formula for $h=\operatorT[f]$ is
  \begin{equation}
    \label{eq:inversion-formula}
    f(x) = f(0) +   \int_0^{x}
    K(x-y) h(y) \dif y,
  \end{equation}
  with (locally integrable) kernel
  \begin{equation}
    \label{eq:T-1_kernel}
    K(x) \doteq - 4\pi \mathrm{i} \int_{\alpha-\mathrm{i}\infty}^{\alpha+\mathrm{i}\infty} \frac{\mathrm{e}^{s x}}{\gamma+\log{s}} \dif s,
    \quad
    \alpha > \mathrm{e}^{-\gamma},
  \end{equation}
  where $\gamma$ is the Euler-Mascheroni constant. Hence, the restriction of $\operatorT$ to $D = \{ f \in C^1[0,r] \mid f(0) = 0 \}$ is such that $\operatorT : D\to D$ and it admits a unique inverse. It is given by
  \begin{equation}
    \label{eq:T-1}
    \operatorT^{-1}[h](x) = \int_0^x K(x-y) h(y) \dif y.
  \end{equation}
The operator $\operatorT^{-1}$ extends to a linear bounded operator on $C[0,r]$ for $r>0$ and
  \begin{gather}
    \label{eq:T-continuity}
    {\|\operatorT^{-1}[h]\|}_\infty \leq C_\infty(r) {\|h\|}_\infty,
  \end{gather}
  where $C_\infty(r) > 0$ depends continuously on $r$ and vanishes in the limit $r \to 0$.
\end{proposition}

\begin{proof}
  We denote by $\mathcal{L}\{f\}$ the Laplace transform of a bounded function $f \in C_b[0,\infty)$,
  \[
    \mathcal{L}\{f\}(s) = \int_0^\infty \mathrm{e}^{-s \tau} f(\tau) \dif \tau.
  \]
  The convolution theorem\footnote{Here $g*f(\tau) = \int_0^\tau g(t) f(\tau-t) \dif t$} for the Laplace transform gives that $\mathcal{L}\{g\}\mathcal{L}\{f\} = \mathcal{L}\{g*f\}$.
  Furthermore, $\mathcal{L}\{f'\}(s) = s\mathcal{L}\{f\}(s) - f(0)$ and thus
  \[
    \mathcal{L}\{\operatorT[f]\}(s) = -\frac{\mathcal{L}(\log)(s)}{8\pi^2} \left( s\mathcal{L}\{f\}(s) - f(0) \right).
  \]
  Let us thus compute
  \[
    \mathcal{L}\{\log\}(s) = \int_0^\infty \mathrm{e}^{-st} \log t \dif t = \int_0^\infty \mathrm{e}^{-st} \log(s t) \dif t - \int_0^\infty \mathrm{e}^{-st} \log s \dif t
    = -\frac{\gamma + \log s}{s}.
  \]
  Hence, we obtain
    \[
    \mathcal{L}\{\operatorT[f]\}(s) = \frac{(\gamma+ \log(s) )}{8\pi^2} \left( \mathcal{L}\{f\}(s)- \frac{f(0)}{s} \right),
    \qquad \mathcal{L}\{f\}(s) =\frac{f(0)}{s} + \frac{8\pi^2}{\gamma+\log(s)}\mathcal{L}\{h\}(s),
  \]
  where $h = \operatorT[f]$. Hence, since  $K$ is the inverse Laplace transform of $8\pi^2(\log(s) + \gamma)^{-1}$, again by the convolution theorem of the Laplace transform we have proved \eqref{eq:inversion-formula}.
  We observe by direct inspection that the restriction of $\operatorT$ to $D$ is closed in $D$. Hence, we also have that the restriction of $\operatorT$ to $D$ admits an unique inverse given in \eqref{eq:T-1}. This finishes the first part of the proof.
  
  For the second part of the proof, note that the function $(\gamma + \log s)^{-1}$ has a simple pole at $s = \mathrm{e}^{-\gamma}$ with residue $\mathrm{e}^{-\gamma}$ and it has a branch cut for $\Re(s) < 0$.
  Furthermore, for $\Re(s) \in (0,\alpha)$ and $x \in \RR$, the function $\mathrm{e}^{sx} (\gamma + \log s)^{-1}$ vanishes in the limit $\Im(s) \to \infty$. Hence, by the Cauchy residue theorem and after a change of variables,
  \begin{align*}
    K(x)  = 8\pi^2 \mathrm{e}^{\mathrm{e}^{-\gamma} x -\gamma} +4\pi\int_{-\infty}^{\infty}   \mathrm{e}^{\mathrm{i} k x}  \frac{1}{\gamma+\log{(\mathrm{i} k)}} \dif k.
  \end{align*}
  To obtain the desired continuity, we need to analyze
  \begin{equation}
    \label{eq:kappa}
    \mathcal{K}(x) \doteq {4\pi} \int_{-\infty}^\infty \frac{\mathrm{e}^{\mathrm{i} k x}}{\gamma + \log(\mathrm{i} k)} \dif k,
  \end{equation}
  which is the Fourier transform of a Schwartz distribution with integral kernel $(\gamma + \log \mathrm{i} k)^{-1}$. The desired continuity follows from the observation that $\mathcal{K}$, given in \eqref{eq:kappa}, is a locally integrable function on $\mathbb{R}$ which is continuous on $\mathbb{R}\setminus\{0\}$ and decays as $|x|^{-1}$ for large $|x|$. The detailed proof of the last statement can be found in the following Lemma \ref{le:kappa} whose proof descends from Lemma \ref{le:I}  and Lemma \ref{le:J}. Hence, thanks to the decay properties of $\mathcal{K}$ and the fact that both $\mathcal{K}$ and thus also $K$ are absolutely integrable near $0$, we have that $\operatorT^{-1}$ extends to a linear bounded operator on $C[0,r]$ for $r>0$ and furthermore
  \[
    {\|\operatorT^{-1}[h]\|}_\infty
    \leq {\|h\|}_\infty \int_0^r |K(x)| \dif x \leq C_\infty(r) {\|h\|}_\infty,
  \]
    where $C_\infty(r)$ depends continuously on $r$ and vanishes in the limit $r\to0$.
\end{proof}

Before introducing the three technical Lemmas used to complete the proof of Proposition \ref{prop:T-inverse} we make the following observations.
Up to the application of a translation, analogous results of Proposition \ref{prop:T-inverse} holds for $\operatorT_{\tau_0}$ and $\operatorT^{-1}_{\tau_0}$.
We observe that an important property of the inversion formula \eqref{eq:inversion-formula} is that it respects causality. Actually $\operatorT^{-1}$ and thus also $\operatorT_{\tau_0}^{-1}$ is a retarded product. 
In the next section, thanks to the continuity shown in Proposition \ref{prop:T-inverse}, we shall be able to prove that a unique solution of \eqref{eq:time-derivative-phi2} exists.

In the proof of the following three Lemmas we use the notation $\log^2 x$ for the square of the logarithm of $x$, i.e., $(\log x)^2$.
Furthermore, the symbol $f \lesssim h$ means that it exists a constant $C$ such that $f\leq Ch$. 

\begin{lemma}\label{le:kappa}
  The function given in \eqref{eq:kappa},
  \begin{equation*}
    \mathcal{K}(x) = 4\pi \int_{-\infty}^\infty \frac{\mathrm{e}^{\mathrm{i} k x}}{\gamma+\log(\mathrm{i} k)} \dif k,
  \end{equation*}
  is continuous for $x \neq 0$, locally integrable near $0$ and bounded outside any interval containing $0$.
\end{lemma}
\begin{proof}
  We divide the $k$-integral in \eqref{eq:kappa} into two parts, obtaining
  \begin{align*}
    \mathcal{K}(x)
    & = 4\pi \int_0^\infty \mleft( \frac{\cos(k x) + \mathrm{i} \sin(kx)}{\gamma + \log k + \mathrm{i}\pi/2} + \frac{\cos(k x) - \mathrm{i} \sin(kx)}{\gamma + \log k - \mathrm{i}\pi/2} \mright) \dif k \\
    & = 8\pi \int_0^\infty \cos(kx) \frac{(\gamma + \log k)}{(\gamma + \log k)^2 + \pi^2/4} \dif k + 4\pi^2 \int_0^\infty \sin(kx) \frac{1}{(\gamma + \log k)^2 + \pi^2/4} \dif k.
  \end{align*}
  We thus observe that the local integrability of $\mathcal{K}$ is equivalent to the local integrability of the functions $I$ and $J$ defined in \eqref{eq:I1} and \eqref{eq:J1}, respectively, after reabsorbing the constant $\gamma$ through the rescalings $k \mapsto k \mathrm{e}^\gamma$ and $x \mapsto x \mathrm{e}^{-\gamma}$.
  The statement then follows from Lemmas \ref{le:I} and \ref{le:J}, where the local integrability and further properties of $I$ and $J$ are established.
\end{proof}

\begin{lemma}
  \label{le:I}
  The function
  \begin{equation}
    \label{eq:I1}
    I(x) \doteq \int_0^\infty \cos(kx) \frac{\log k}{\log^2 k + c} \dif k,
    \quad c > 0,
  \end{equation}
  is a continuous functions for $x \neq 0$, it is locally integrable near $0$ and it decays as $|x|^{-1}$ for large $|x|$.
\end{lemma}
\begin{proof}
  As the following proof can be easily generalized to arbitrary $c > 0$, we study only the case $c = 1$.
  To prove continuity outside $0$ and the decay for large values of $x$, we integrate by parts obtaining
  \begin{equation}
    \label{eq:I2}
    I(x) = \int_0^\infty \frac{\sin(kx)}{x} \frac{\log^2 k - 1}{k (\log^2 k + 1)^2} \dif k.
  \end{equation}
  Hence we have that
  \[
    |I(x)| \leq \frac{1}{|x|} \int_0^\infty \frac{|\log^2 k - 1|}{k (\log^2 k + 1)^2} \dif k = \frac{1}{|x|} \int_{-\infty}^\infty \frac{|l^2-1|}{(l^2+1)^2} \dif l = \frac{2}{|x|}
  \]
  and continuity can be proved by the dominated convergence theorem. 

  In order to prove the integrability near $x=0$, we integrate by parts another time in \eqref{eq:I2} to obtain
  \[
    I(x) = \int_0^\infty \frac{1-\cos(kx)}{x^2} \frac{1}{k^2} \mleft( \frac{\log^2 k - 2\log k - 1}{(\log^2 k + 1)^2} + \frac{4 \log(k) (\log^2 k - 1)}{(\log^2 k + 1)^3} \mright) \dif k,
  \]
  and assume that $x \in (0, \varepsilon)$ for $\varepsilon$ sufficiently small (the case $x<0$ can be treated analogously).
  After changing the variable of integration ($k \mapsto k x$), we get
  \[
    I(x) = \frac{1}{x} \int_0^\infty \frac{1-\cos k}{k^2} \frac{f(l)}{l^2+1} \dif k,
  \]
  where $l \doteq \log k - \log x$ and \[ f(l) \doteq \frac{l^2-2l-1}{l^2+1} + \frac{4l(l^2-1)}{(l^2+1)^2} \] is a continuous bounded function.

  We split the integral into two parts $I(x) = I_1(x)+I_2(x)$, where
  \[
    I_1(x) \doteq \frac{1}{x} \int_0^{\sqrt{x}} \frac{1-\cos k}{k^2} \frac{f(l)}{l^2+1} \dif k,
    \quad
    I_2(x) \doteq \frac{1}{x} \int_{\sqrt{x}}^\infty \frac{1-\cos k}{k^2} \frac{f(l)}{l^2+1} \dif k,
  \]
  and discuss local integrability near $0$ separately for $I_1$ and $I_2$.
  Since $(l^2+1)^{-1} \leq 1$, $|1-\cos k| \leq \frac{k^2}{2}$ and $|f(l)| \lesssim 1$,
  \[
    |I_1(x)| \lesssim \frac{1}{x} \int_0^{\sqrt{x}} \frac{1-\cos k}{k^2}\dif k \lesssim \frac{1}{\sqrt x},
  \]
  which is integrable in $(0, \varepsilon)$.
  At the same time we have that
  \[
    |I_2(x)| \lesssim \frac{1}{x \log^2 x} \int_{\sqrt{x}}^\infty \frac{1-\cos k}{k^2} \dif k \leq \frac{1}{x \log^2 x} \int_0^\infty \frac{1-\cos k}{k^2} \dif k \lesssim \frac{1}{x \log^2 x},
  \]
  where we used $|f(l)| \lesssim 1$ and
  \[
    (\log k - \log x)^2 + 1 > (\log k - \log x)^2 \geq \frac{1}{4} \log^2 x
  \]
  because $\log x < 0$ and, on the domain of $k$-integration, $\log k \geq \frac{1}{2} \log x$.
  Consequently also $I_2$ is integrable in $(0, \varepsilon)$ for small $\varepsilon$.
\end{proof}

\begin{lemma}
  \label{le:J}
  The function
  \begin{equation}
    \label{eq:J1}
    J(x) \doteq \int_0^\infty \sin(kx) \frac{1}{\log^2 k + c} \dif k,
    \quad c > 0,
  \end{equation}
  is a continuous functions for $x \neq 0$, it is locally integrable near $0$ and it decays as $|x|^{-1}$ for large $|x|$.
\end{lemma}
\begin{proof}
  We proceed similarly as in the previous proof, studying only the case $c = 1$.
  We integrate by parts to obtain
  \[
    J(x) = \int_0^\infty \frac{1-\cos(kx)}{x} \frac{2 \log k}{k (\log^2 k + 1)^2} \dif k
  \]
  and thus
  \[
    |J(x)| \leq \frac{4}{|x|} \int_0^\infty \frac{|\log k|}{k (\log^2 k + 1)^2} \dif k = \frac{4}{|x|} \int_{-\infty}^\infty \frac{|l|}{(l^2 + 1)^2} \dif l = \frac{4}{|x|}.
  \]
  Continuity of $J(x)$ for $x \neq 0$ can then be proved by the dominated convergence theorem.

  We are left with the proof of integrability near $x = 0$.
  For this purpose we assume that $x \in (0, \varepsilon)$ for $\varepsilon$ sufficiently small (the case $x < 0$ can be treated analogously).
  Dividing the domain of integration of $J(x)$ into $(0,1/\sqrt{x})$ and $(1/\sqrt{x},\infty)$, we obtain $J(x) = J_1(x) + J_2(x)$ with
  \begin{align*}
    J_1(x) & = \int_0^{1/\sqrt{x}} \frac{1-\cos(kx)}{x} \frac{2 \log k}{k (\log^2 k + 1)^2} \dif k, \\
    J_2(x) & = \int_{1/\sqrt{x}}^\infty \frac{1-\cos(kx)}{x} \frac{2 \log k}{k (\log^2 k + 1)^2} \dif k.
  \end{align*}
  Noting that $|1-\cos(kx)| \leq \frac12 (kx)^2$ and $\log(k) (\log^2 k + 1)^{-2} \lesssim 1$, we estimate $J_1$ as
  \[
    |J_1(x)| \lesssim \int_0^{1/\sqrt{x}} \frac{1-\cos(kx)}{kx} \dif k \leq \int_0^{1/\sqrt{x}} \frac{kx}{2} \dif k = \frac14.
  \]
  At the same time we find for $J_2$ that
  \[
    |J_2(x)| \leq \frac{4}{x} \int_{1/\sqrt{x}}^\infty \frac{|\log k|}{k (\log^2 k + 1)^2} \dif k = \frac{4}{x} \int_{-\frac12 \log x}^\infty \frac{|l|}{(l^2 + 1)^2} \dif l = \frac{8}{4x + x \log^2 x}.
  \]
  Hence, we can conclude that $J$ is integrable in $(0, \varepsilon)$ for small $\varepsilon$.
\end{proof}

\subsection{Existence and uniqueness of local weak solutions}\label{se:sol}

In this section we shall present the main result of this paper, namely the existence and uniqueness of solutions of the semiclassical Einstein equation \eqref{SCE-trace} for a fixed arbitrary coupling parameter $\xi\neq 1/6$.
We shall use all the results previously obtained in order to translate the original semiclassical equation in the form given in \eqref{eq:time-derivative-phi2} into an of the form \eqref{fixed-point-equation}. We shall use the continuity property of the inverse operator $\operatorT_0^{-1}$ given in \eqref{eq:T-1} and proved in Proposition \ref{prop:T-inverse} in order to define a suitable contraction map.

We preliminarily observe that in an interval of time $[\tau_0,\tau]$ it is possible to control $a,{a'}$ and $V$ by means of $X={a''}/a$ and thus by means of $X'$.
See Lemma \ref{le:a[x]} given in the appendix for further details.
In particular, we have that \[ \|V\|_\infty \leq C \left(1+\|X-X_0\|_\infty \right), \]  where $X_0=X(\tau_0)$ and $C$ is a suitable constant. The first step is to rewrite \eqref{SCE-V} in terms of the dynamic variable $X'$, in order to obtain the explicit expression of the map $\cC$.
\begin{lemma}
  \label{lem:cT}
  Given the initial data $(a_0,{a'}_0,X_0,{X}'_0)$, chosen in such a way that $a_0>0$ and $\Omega_k^{2}(\tau_0)$ in \eqref{mode-eq} is strictly positive, and a state $\omega$ which is regular and compatible with this initial conditions, the semiclassical equation \eqref{SCE-V} can be written in the form of a fixed-point equation on $C[\tau_0,\tau_1]$
  \begin{equation}
    \label{SCE-fixed-point}
    X'=\mathcal{C}[X'],
  \end{equation}
  where
  \begin{align*}
    \mathcal{C}[X']  = & X'_0 -  \frac{2m^2}{(6\xi-1)}(a[X]a'[X]- a_0a_0') \\  & - \frac{1}{(6\xi-1)}\operatorT^{-1}_{\tau_0}\left[Q^d_0[X] + Q_f^d[X] + Q^d_s[X]   -  \left( 6 c_\xi X' + \partial_\tau \left(a[X]^2 {F}(a[X],R[X]) \right)\right)\right],
  \end{align*}
  with $X[X'](\tau) = X_0 + \int^\tau_{\tau_0} X'(\eta) \dif \eta$. Each $X' \in  C[\tau_0,\tau_1]$ determines a spacetime $(\cM,g[X])$ where $\cM = [\tau_0,\tau_1] \times \mathds{R}^3$ and where $g[X]$ is the FLRW metric enjoying the initial conditions and constructed out of the scale factor $a[X] (\tau)$.
\end{lemma}
\begin{proof}
  Each $X'\in C[\tau_0,\tau_1]$ determines a FLRW spacetime in the following way.
  First of all $X$ is obtained from $X'$ integrating in time and fixing $X(\tau_0)=X_0$. Thus, $X$ is a functional of $X'$. Then, $a[X]$ is obtained from $X$ as the unique solution of $a''-Xa = 0$ which satisfies the initial conditions $a(\tau_0)=a_0$ and $a'(\tau_0)=a'_0$.
  Finally, knowing $X$ and $a$ we observe that $V$ can be obtained from \eqref{potential-2}.
  It is thus a functional of $X$ and hence of $X'$.

  Equations \eqref{eq:time-derivative-phi2} and \eqref{SCE-V} have the form \[ \operatorT_{\tau_0}[V'] = h, \] where \[ h = -Q^d_0[X] - Q_f^d[X] - Q^d_s[X]   +  \left( 6 c_\xi X' + \partial_\tau \left(a[X]^2 {F}(a[X],R[X]) \right)\right).
  \]
  We invert this equation adapting the analysis given in Proposition \ref{prop:T-inverse} to obtain
  \[
    V'=V'_0+\operatorT_{\tau_0}^{-1}[h].
  \]
  The operator $\operatorT_{\tau_0}^{-1}$ equals $\operatorT_0^{-1}$ given in \eqref{eq:T-1} up to a translation, furthermore, the continuity satisfied by $\operatorT_{\tau_0}^{-1}$ coincides with the continuity of $\operatorT_0^{-1}$ discussed in equation \eqref{eq:T-continuity} of Proposition \ref{prop:T-inverse}. Finally, we rewrite it with respect to the variable $X = \frac{{a''}}{a}$.
\end{proof}

Once the semiclassical Einstein equation is given in this form, we can prove the following
\begin{proposition}
  \label{prop:contraction}
  Fix the initial conditions for $a$ in such a way that $a_0>0$ and $\Omega_k^2(\tau_0)$ given in \eqref{mode-eq} is positive. Consider a state $\omega$ which is sufficiently regular and compatible with the initial conditions for $a$. 
  Fix $\delta >0$ and let $\cB_\delta$ given in \eqref{eq:ball-delta} the closed ball in the Banach space $C[\tau_0,\tau_1]$  with finite $\tau_1 > \tau_0$, centred in $X_c'(\tau) \doteq X'_0$.
  For $\tau_1$ sufficiently small, the map $\mathcal{C}$ introduced in Lemma \ref{lem:cT} with $\xi\neq 1/6$ is a contraction map on $\cB_{\delta}$.
  Hence, there exists a unique fixed point of the equation $X'=\mathcal{C}[X']$, in $\mathcal{B}_\delta$, which represents a solution of the semiclassical Einstein equation.
\end{proposition}
\begin{proof}
  First of all, we observe that for every $X'\in \mathcal{B}_\delta$ we assign an $X(\tau)= X_0 +\int_{\tau_0}^\tau X'(\eta)\dif \eta$ and consequently a scale factor $a[X]$. $\Omega_k^2$ given in \eqref{mode-eq} is continuous in time, since it is positive at $\tau_0$, it stays positive in a short interval of time. Furthermore, for $X'\in \mathcal{B}_\delta$, $|a^2R(\tau)| \leq (\tau-\tau_0)\|X'\|_\infty$, hence if $\tau_1$ is sufficiently small $\Omega_k^2$ is positive in $[\tau_0,\tau_1]$ uniformly for $X'\in\mathcal{B}_\delta$. 
  
  The strategy of the proof is the following. We observe that $\mathcal{C}$ is a linear combination of compositions of continuous functions or functionals of $a$, $V$ and $X$. $a$ and $V$ are Gateaux differentiable with respect to $X$ at $X(\tau)= X_0 +\int_{\tau_0}^\tau X'(\eta)\dif \eta$, furthermore their derivative satisfy the inequalities derived in Lemma \ref{le:a[x]}, thus all these Gateaux differential are continuous.
  Hence, the proof of this proposition follows from the continuity of $\operatorT_{\tau_0}^{-1}$ given in Proposition
  \ref{prop:T-inverse} and observing that if $(\tau_1-\tau_0)$ tends to $0$ then the constant $C_\infty(\tau_1-\tau_0)$ given in \eqref{eq:T-continuity} tends to $0$.
  The thesis follows from the continuity of all the other operators, functionals or functions involved.

  To fix some details of the proof we proceed as follows. Consider the following constants $c_1= -2m^2/(6\xi-1)$ and $c_2= (6\xi-1)^{-1}$, the map $\mathcal{C}$ has the form
  \[
    X'= \mathcal{C}[X']= X_0' + c_1 (a-a_0)a' + c_1 a_0(a' -a_0') - c_2 \operatorT^{-1}_{\tau_0}\left[\mathcal{F}\right],
  \]
  where
  \begin{equation}\label{eq:def-F}
    \mathcal{F}= Q^d_0 + Q_f^d + Q^d_s -  \left( {6 c_\xi X'} + \partial_\tau \left(a^2 {F}(a,R) \right)\right)
  \end{equation}
  is a linear combination of continuous functionals of $X$. If $X \in \mathcal{B}_\delta$, we have that
  \begin{equation}
  \label{eq:con1-a}
    \begin{aligned}
      \|\mathcal{C}[X']-X_0' \|_\infty & \leq |c_1|  ( \|a-a_0\|_\infty \|a'\|_\infty + a_0 \|a'-a_0'\|_\infty) + |c_2| C_\infty(\tau_1-\tau_0) \| \mathcal{F}\|_\infty,
    \end{aligned}
  \end{equation}
  where we used the estimates given in Lemma \ref{le:a[x]} and in Proposition \ref{prop:T-inverse}. We now observe that for $X'\in \mathcal{B}_{\delta}$, it is possible to prove that $\mathcal{F}$ is bounded by $X'$. More precisely, using the results obtained above we can bound every component of $\mathcal{F}$ given in \eqref{eq:def-F}. We have actually established in Proposition \ref{prop:decomposition} that $Q^d_f$ depends continuously on $V'$ and in Proposition \ref{prop:FV} that $\|Q^d_s\|_\infty$ can be controlled by $\|V'\|_\infty$ for $V'$ in a suitable compact domain of $C[\tau_0,\tau_1]$ because $Q^d_s$ is Gateaux differentiable. At the same time $Q^d_0$ given in \eqref{eq:Qd0} is a function of $a$ and its derivative up to the third order and of $V$ and $V'$.   
  We also have that $V'$ depends continuously on $X'$ with respect to the uniform topology as can be seen from the Definition of $V$, see e.g. \eqref{potential-2}, and thanks to the results of Lemma \ref{le:a[x]}. 
  Furthermore, as established in Theorem \ref{theo:sol-SCE-wave}, $F$ is a solution of the first equation in \eqref{SCE-wave} and hence it can be controlled by $X'$ again together with its time derivative, similarly to the results established in Proposition \ref{prop:Pc}. Combining all these observations we have that if $\tau_1$ is chosen sufficiently small, it exists $C_\delta$ which permits to further bound the right hand side of \eqref{eq:con1-a}, hence we have that
  \begin{equation}
  \label{eq:con1}
      \|\mathcal{C}[X']-X_0' \|_\infty  \leq \left( (\tau_1-\tau_0) + |c_2| C_\infty(\tau_1-\tau_0)   \right) C_\delta.
  \end{equation}
  In particular we recall that the constant $C_\infty(\tau_1-\tau_0)$ depends continuously on the difference $\tau_1 - \tau_0$ and it vanishes for $\tau_1=\tau_0$.
  
  Furthermore, for $X_1,X_2\in \mathcal{B}_\delta$,
  \[
    \mathcal{C}[X'_2]-\mathcal{C}[X'_1] =
    c_1 (aa'[X_2]-aa'[X_1]) - c_2 \operatorT^{-1}_{\tau_0}\left[\mathcal{F}[X_2]-\mathcal{F}[X_1]\right]
  \]
  and
  \[
    \|\mathcal{C}[X'_2]-\mathcal{C}[X'_1]\|_\infty \leq
    |c_1| \| aa'[X_2]-aa'[X_1]\|_\infty + |c_2| C_\infty(\tau_2-\tau_0) \| \mathcal{F}[X_2]-\mathcal{F}[X_1]\|_\infty.
  \]
  Considering the convex linear combination of $X_1$ and $X_2$, $X_s = (1-s) X_1 + s X_2 = X_1 + s(\delta X)$ where $\delta X = X_2-X_1$, using the definition of directional derivative, see Remark \ref{re:functionalderivative}, we have that
  \[
    \mathcal{F}[X_2]-\mathcal{F}[X_1] =  \int_0^1     \frac{\dif  \mathcal{F}[X_s] }{\dif s}  \dif s = \int_0^1  \delta\mathcal{F}[X_s,\delta X] \dif s.
  \]
  The space $\mathcal{B}_\delta$ is a convex space, hence $X_s\in\mathcal{B}_\delta $ for every $s$.
  To control the functional derivative of $\mathcal{F}$ given in \eqref{eq:def-F}, we analyze the functional derivatives of its component. 
  In view of Remark \ref{re:functionalderivative} and having control on how $V$ depends on $X$ and on $a$, see e.g. \eqref{potential-2}, the boundedness of the functional derivative of $Q^d_f$ with respect to $X'$ descends from Lemma \ref{le:a[x]} and from the bounds established in Proposition \ref{prop:decomposition}. Similarly, the boundedness of $Q^d_s$ descends from Proposition \ref{prop:FV} and that of $Q^d_0$ directly from its Definition in \eqref{eq:Qd0}. Finally, the boundedness of the functional derivative of the time derivative of $a^2 F$ can be obtained arguing as in Proposition \ref{prop:Pc}, see also the explicit results stated in Theorem \ref{theo:sol-SCE-wave}.
  Collecting all these observations, we have that, for $X_1,X_2\in \mathcal{B}_\delta$, there exists a suitable constant ${C}$ such that $\| \delta\mathcal{F}[X_s,\delta X] \|_{\infty} \leq  {C}\|\delta X'\|_\infty$. Hence, 
  \[
    \| \mathcal{F}[X_2]-\mathcal{F}[X_1] \|_\infty \leq C \|X_2' - X_1'\|_\infty.
  \]
  Furthermore, operating in a similar way for the first contribution in the difference $\mathcal{C}[X'_2]-\mathcal{C}[X'_1]$, we have
  \begin{align*}
    aa'[X_2]-aa'[X_1]
    &= a[X_2](a'[X_2]-a'[X_1])+(a[X_2]-a[X_1]) a'[X_1] \\
    &= a[X_2] \int_0^1 \dif s \delta a'[X_s,\delta X] + a'[X_1] \int_0^1 \dif s \delta a[X_s,\delta X],
  \end{align*}
  where $\delta a$ is the functional derivative of $a$. Hence, using estimates similar to those of Lemma \ref{le:a[x]}, we get
  \[
    \| aa'[X_2]-aa'[X_1] \|_\infty \leq (\tau_1-\tau_0) C \|X_2'-X_1'\|_\infty,
  \]
  where $C$ is a suitable constant.
  Combining these results and using the continuity of $\operatorT_{\tau_0}^{-1}$ obtained in Proposition \ref{prop:T-inverse}, we have that for a suitable constant $C$ which does not depend on $\tau_1-\tau_0$ for $\tau_1-\tau_0< \epsilon$:
  \begin{equation}
    \label{eq:con2}
    \| \mathcal{C}[X'_2]-\mathcal{C}[X'_1]\|_\infty \leq ((\tau_1-\tau_0) + C_\infty(\tau_1-\tau_0)) C \|X'_2-X'_1\|_\infty.
  \end{equation}
  We thus have that for $\tau_1-\tau_0$ sufficiently small the action of $\mathcal{C}$ is internal in $\mathcal{B}_\delta$ thanks to \eqref{eq:con1} and at the same time $\mathcal{C}$ is a contraction map thanks to \eqref{eq:con2}.
\end{proof}

\begin{theorem}
  \label{theo:contractio}
  Let $(a_0, {a'}_0, {a''}_0, a^{(3)}_0)$ be some initial data for the functional equation \eqref{fixed-point-equation} given at $\tau_0$ with $a_0>0$ and such that $\Omega_k^2(\tau_0)$ given in \eqref{mode-eq} is strictly positive.
  Consider a quasifree state $\omega$, which is sufficiently regular and
  compatible with these initial conditions. There exist a non-empty interval $[\tau_0,\tau_1]$ and a closed ball $\cB_\delta =\{X'\in C[\tau_0,\tau_1] \mid X'(\tau_0)=X'_0\}$ of radius $\delta >0$ such that, for sufficiently small $\tau_1$, a unique solution to \eqref{fixed-point-equation} exists.
\end{theorem}
\begin{proof}
  The existence of a regular quasifree state compatible with the initial conditions for $a$ is established in Proposition \ref{prop:Friedman-constraint}.
  On account of Proposition \eqref{prop:contraction}, the proof is an application of the Banach fixed point theorem to the contraction map $\mathcal{C}$ on $\cB_{\delta}$.
\end{proof}

\begin{remark}
  The scale factor $a$ corresponding to the unique solution obtained in Theorem \ref{theo:contractio} is an element on $C^3[\tau_0,\tau_1]$.
  We do not have direct control on its fourth order derivative.
  Having third order derivative of $a$ at disposal, we can thus directly check the validity of the first Friedmann equation at any time in $[\tau_0,\tau_1]$, but that regularity is not sufficient to control the traced semiclassical Einstein equation in the form \eqref{trace-equation} at $\tau$ larger than $\tau_0$.
  For this reason, the obtained solution is only a mild solution of the semiclassical problem.
  To improve this result, namely to obtain a unique solution $a \in C^4[\tau_0,\tau_2]$ for some $\tau_0<\tau_2<\tau_1$, there is the need of a better control on the state.
  But this could be achieved imposing constraints on the initial conditions of the modes in \eqref{eq:init-cond-state} stronger than those given in \eqref{eq:init-regular1} and in \eqref{eq:init-regular2}.
\end{remark}

\begin{remark}
  Combining the results of Proposition \ref{prop:Friedman-constraint}, of Theorem \ref{prop:Friedman-constraint} and of Theorem \ref{theo:contractio} we have proved that it is possible to find a unique solution of the semiclassical Einstein equation.
  The control on $\langle\wick{\phi^2}\rangle_\omega$ provided by Proposition \ref{prop:decomposition} and the analysis about the continuity of $\operatorT_{\tau_0}^{-1}$ yield the continuity of the obtained solution with respect to the initial conditions for the scale factor.
  Actually, $\operatorT_{\tau_0}^{-1}$ does not depend on the initial conditions.
  The unique solution $F$ obtained in \ref{prop:Friedman-constraint} depends continuously on its initial data and the estimates of Proposition \ref{prop:Friedman-constraint} permit to control how the initial data for $F$ depend on the initial data of the scale factor.
\end{remark}

\section{Conclusion}

In this paper we have studied the backreaction of a quantum linear scalar field coupled with gravity on cosmological spacetimes.
We have shown that a unique solution on a small interval of time exists once some initial conditions at finite time $\tau = \tau_0$ are fixed.
Having established the existence and uniqueness of solutions, it is now meaningful to look for numerical algorithms to find approximate solutions.
However, as we have seen in this paper, in order to have a meaningful fixed point equation, the semiclassical equation needs to be rewritten in a non-standard form and only after this step it is possible to apply the Banach fixed point theorem.
Hence, to find numerical solutions, a possibility is to recursively apply the contraction map $\mathcal{C}$ to some initial spacetime since the convergence of this methods is thus guaranteed.
A recursive procedure obtained by a direct application of the semiclassical equation will hardly be convergent because of the loss of derivatives present in the expectation values of the field observables.

There are still open questions  in particular a discussion about the existence and uniqueness of global (maximal) solutions, as carried out in \cite{Pinamonti:2013wya} and \cite{Gottschalk:2018kqt}, is missing. In this framework, Ostrogradsky's theorem merits a remark, since higher-order derivative terms in $\expvalom{T_{ab}}$ could represent a source of instability inside the equations. In \cite{Koksma:2008jn} for instance, this problem is pointed out referring to the trace-anomaly term. For an outline about Ostrogradsky's theorem, see for instance \cite{Woodard:2006nt, Woodard:2015zca}.
A discussion about the limits of validity of the solutions of semiclassical equations and the role played by their non-classical terms is present in \cite{Flanagan:1996gw}. Moreover, a ``reduction of order'' prescription to select physically reasonable solutions is proposed, following the so-called reduced Simon-Parker theory \cite{Simon:1991red, ParkerSimon:1993} (see also \cite{SiemieniecOzieblo:1999fz}). As already remarked in \cite{Gottschalk:2018kqt}, it could be interesting to investigate how that prescription could be applied after rewriting the semiclassical equations in these non-standard forms.

A prime generalization of our analysis can be carried out on non-flat Robertson-Walker spacetimes. We expect that the same techniques can be adopted to study the system of equations \eqref{SCE-wave-intro} even in this case, after imposing similar conditions for a sufficiently regular state (actually, adiabatic states can be constructed also for this class of spacetimes). 

Finally, we expect that the problem with the higher derivatives is present also for other different choices of backgrounds, e.g. spherically symmetric spacetimes. We thus expect that also in the analysis of black hole evaporation on four-dimensional spacetimes the semiclassical equations need to be rewritten in an appropriate way before looking for approximate solutions (the two-dimensional case is well-discussed in \cite{Ashtekar:2010qz, Ashtekar:2010hx}).

\subsection*{Acknowledgments}
We would like to thank Hanno Gottschalk for the many discussions we had on the problem of the semiclassical Einstein equation in cosmology and for his careful reading of an earlier version of this paper.

\appendix

\section{Second order differential equations}

In the text we have often obtained equations for $f\in C^{n}[\tau_0,\infty)$ with $n\geq 2$ of the form
\begin{equation}\label{eq:second-order-ode}
\begin{cases}
	f''+ (k^2+ W)f = h, \\
	(f(\tau_0),f'(\tau_0))  = (f_0,f'_0), 
\end{cases}
\end{equation}
where $k$ is some constant, $W$ and $h$ are known functions in $C^{n}[\tau_0,\infty)$ and where $f_0,f_0'$ are suitable constants expressing initial conditions for $f$ at $\tau_0$. 
By standard results we know that a unique solution $f$ of \eqref{eq:second-order-ode} exists. 
In the next lemma we derive some useful properties of the solution of \eqref{eq:second-order-ode}.

\begin{lemma}
  \label{le:second-order-ODE}
  Let $f\in C^n[\tau_0,\infty)$ be the unique solution of \eqref{eq:second-order-ode}. Hence, for $k\geq 0$
  \begin{equation}
  \label{eq:rewert-equation}
 	{f}  =  - \Delta_R^k* ({W} {f})   +  \Delta_R^0*   {h} + {f}_0 \cos(k(\tau-\tau_0)) + {f}'_0 \frac{\sin(k(\tau-\tau_0))}{k},
  \end{equation}
  where $\Delta_R^k(\tau) = \frac{\sin(k\tau)}{k} \theta(\tau)$ for $k\geq0$ is the retarded fundamental solutions of $\dif^2/\dif\tau^2 +k^2$ and in particular at $k=0$ $\Delta_R^0(\tau) = \tau \theta(\tau)$. Furthermore, the convolution $*$ is computed on the interval $[\tau_0,\infty)$. Then, the following estimate holds for $k\geq 0$, $\tau\geq \tau_0$:
  \begin{equation}
    \label{eq:second-order-ODE-ineq}
    |{f}(\tau)| \leq \left( |f_0| + (\tau-\tau_0)|f'_0|+ (\tau-\tau_0)^2 \|{h}\|_\infty \right)\exp \left( (\tau-\tau_0)^2\| {W} \|_\infty\right).
  \end{equation}
    Furthermore, for $k>0$ we have for $\tau\geq \tau_0$
  \begin{equation}
  	\label{eq:second-order-ODE-ineq-k}
  	|f(\tau)| \leq \left( |f_0| + \frac{1}{k}|f'_0|+ \frac{1}{k} \int_{\tau_0}^\tau  |{h}(\eta)| \dif \eta  \right)\exp    
   \left( \frac{1}{k}\int_{\tau_0}^\tau  |W|  \dif \eta \right).
  \end{equation}
\end{lemma}

\begin{proof}
Equation \eqref{eq:rewert-equation} can be obtained computing the convolution of both sides of 
\eqref{eq:second-order-ode} with $\Delta_R^k$ on $[\tau_0,\infty)$ and integrating by parts a couple of times.
We obtain the desired estimates applying Gr\"onwall lemma in the form which states that, if $u(\tau)\leq \alpha(\tau) +\int_{\tau_0}^\tau \beta(\eta) u(\eta) \dif \eta  $ for $\beta$ a non negative function on $[\tau_0,\infty)$ and $\alpha$ a non decreasing function on  $[\tau_0,\infty)$, then it holds that $u(\tau) \leq \alpha(\tau) \exp \int_{\tau_0}^\tau \beta(\eta) \dif \eta$.
In particular, from equation \eqref{eq:rewert-equation} we get for $\tau\geq \tau_0$
\[
  |f(\tau)|  \leq  (\tau-\tau_0) \int_{\tau_0}^\tau | {W}(\eta) {f}(\eta)| \dif \eta + (\tau-\tau_0) \int_{\tau_0}^\tau   |{h}(\eta)| \dif \eta + |{f}_0| + |{f}'_0| (\tau-\tau_0),  
\]
or for $k>0$
\[
  |f(\tau)| = \frac{1}{k} \int_{\tau_0}^\tau |{W}(\eta) {f}(\eta)| \dif \eta + \frac{1}{k} \int_{\tau_0}^\tau |{h}(\eta)| \dif \eta + |{f}_0| + \frac{1}{k}  |{f}'_0|, 
\]
hence by Gr\"onwall lemma we get the desired estimate for $|f(\tau)|$ stated in \eqref{eq:second-order-ODE-ineq} and 
in \eqref{eq:second-order-ODE-ineq-k}.
\end{proof}

\begin{lemma}
  \label{le:a[x]}
  Let $a\in C^{2}[\tau_0,\tau_1]$ be the unique solution of $a'' = Xa$ with $a'(\tau_0)=a'_0$ and $a(\tau_0)=a_0$ with $X\in C^{1}[\tau_0,\tau_1]$.
  Then the following inequalities hold:
  \begin{align*}
    \|a - a_0\|_\infty    & \leq (\tau_1-\tau_0) \left(   |a_0'| + |a_0| \frac{(\tau_1-\tau_0)}{2} \|X\|_\infty \right) \exp \left(\frac{(\tau_1-\tau_0)^2}{2} \|X\|_\infty \right), \\ \| a'-a_0'\|_\infty & \leq \frac{(\tau_1-\tau_0)^2}{2} \left(a_0\|X\|_\infty + \|a\|_\infty \|X'\|_\infty \right) \exp \left(\frac{(\tau_1-\tau_0)^2}{2} \|X\|_\infty \right),
    \\
    \| \delta a \|_\infty & \leq \frac{(\tau_1-\tau_0)^2}{2}  \| a\|_\infty
    \exp \left(\frac{(\tau_1-\tau_0)^2}{2} \|X\|_\infty \right)\|\delta X\|_\infty,
    \\
    \|\delta a'\|_\infty  & \leq \frac{(\tau_1-\tau_0)^2}{2}\left(\| (a\delta X)'\|_\infty+\| X'\delta a \|_\infty \right)
    \exp \left(\frac{(\tau_1-\tau_0)^2}{2} \|X\|_\infty \right),
  \end{align*}
  where $\delta a[X,\delta X]$ denotes the functional derivatives with respect to infinitesimal changes $\delta X \in C^{1}[\tau_0,\tau_1]$ and where the uniform norms are computed on the interval $(\tau_0,\tau_1)$.
\end{lemma}

\begin{proof}
We apply Lemme \eqref{le:second-order-ODE} to the equation $a''=Xa$, namely for $k=0$. In that case the retarded fundamental solution $\Delta_R(\tau) =\tau\theta(\tau)$, hence from \eqref{eq:rewert-equation} we get 
  \begin{equation}
    \label{eq:a(tau)}
    a(\tau) = a_0+(\tau-\tau_0) a_0'+ \int_{\tau_0}^\tau (\tau-\eta) X(\eta )a(\eta) \dif \eta.
  \end{equation}
  We have that
  \[
  	|a - a_0| \leq (\tau-\tau_0) |a_0'| + \int_{\tau_0}^\tau (\tau-\eta) | X(\eta )||a(\eta)-a_0|\dif \eta + |a_0| \frac{|(\tau-\tau_0)|^2}{2} \|X\|_\infty.
  \]
  Gr\"onwall inequality gives the first inequality. The bound for the first derivative can be obtained in a similar way starting from the first derivative of equation \eqref{eq:a(tau)}
  \[
	  a'(\tau)-a_0' =  \int_{\tau_0}^\tau (\tau-\eta) \left(X(\eta )a_0+X'(\eta )a(\eta) \right) \dif \eta + \int_{\tau_0}^\tau (\tau-\eta) X(\eta )(a'(\eta)-a'_0) \dif \eta, 
  \] 
  then writing the corresponding local inequality and finally applying again Gr\"onwall inequality. The inequalities for the functional derivatives are obtained computing the first functional derivatives of \eqref{eq:a(tau)}, which read
  \begin{align*}
    \delta a & = \int_{\tau_0}^\tau (\tau-\eta) \left( \delta X(\eta )a(\eta) + X(\eta )\delta a(\eta) \right)  \dif \eta, \\ \delta a' & = \int_{\tau_0}^\tau (\tau-\eta) \left( \delta X(\eta )a(\eta))' + X'(\eta )\delta a(\eta)+ X(\eta )\delta a'(\eta) \right)  \dif \eta.
  \end{align*}
  Eventually, we get the desired results after operating as before.
\end{proof}

\newcommand{\SortNoop}[1]{}

\end{document}